\documentclass[11pt]{article}
\usepackage{asherspector}

\usepackage[margin=0.9in]{geometry}
\usepackage{color}
\usepackage{graphicx}
\usepackage[font=small,skip=0pt]{caption}
\usepackage{subcaption}
\usepackage{multirow}

\usepackage{rotating}
\usepackage{verbatim}
\usepackage{hyperref}
\usepackage{bm}
\usepackage[normalem]{ulem}
\usepackage[authoryear]{natbib}
\usepackage{bbm}
\usepackage{array}
\usepackage{float}
\usepackage{indentfirst}
\usepackage{algorithm}
\usepackage[noend]{algpseudocode}

\usepackage{boldline}

\usepackage[flushleft]{threeparttable}
\usepackage{tablefootnote}

\usepackage{array}
\newcommand{\PreserveBackslash}[1]{\let\temp=\\#1\let\\=\temp}
\newcolumntype{C}[1]{>{\PreserveBackslash\centering}p{#1}}

\newcommand{\pval}{p_\mathrm{val}}
\newcommand{\ols}{^{\mathrm{ols}}}

\newcommand{\multcust}[2]{\the\numexpr#1*#2\relax}

\newcommand{\stardate}{May 21st, 2020}

\usepackage[perpage]{footmisc}

\newlist{steps}{enumerate}{1}
\setlist[steps, 1]{label = Step \arabic*:}
\newcommand{\mosaicpermalgnum}{1}

\numberwithin{equation}{section}

\title{The mosaic permutation test: an exact and nonparametric goodness-of-fit test for factor models}

\author{Asher Spector \thanks{Department of Statistics, Stanford University} \and Rina Foygel Barber \thanks{Department of Statistics, University of Chicago} \and Trevor Hastie \footnotemark[1] \and Ronald N. Kahn \thanks{BlackRock, Systematic	
Investment	Research} \and Emmanuel  Cand{è}s \footnotemark[1] \thanks{Department of Mathematics, Stanford University}
}

\begin{document}
\maketitle

\begin{abstract}
Financial firms often rely on fundamental factor models to explain correlations among asset returns and manage risk. Yet after major events, e.g., COVID-19, analysts may reassess whether existing risk models continue to fit well: specifically, after accounting for a set of known factor exposures, are the residuals of the asset returns independent? With this motivation, we introduce the mosaic permutation test, a nonparametric goodness-of-fit test for preexisting factor models. Our method can leverage modern machine learning techniques to detect model violations while provably controlling the false positive rate, i.e., the probability of rejecting a well-fitting model, without making asymptotic approximations or parametric assumptions. This property helps prevent analysts from unnecessarily rebuilding accurate models, which can waste resources and increase risk. To illustrate our methodology, we apply the mosaic permutation test to the BlackRock Fundamental Equity Risk (BFRE) model. Although the BFRE model generally explains the most significant correlations among assets, we find evidence of unexplained correlations among certain real estate stocks, and we show that adding new factors improves model fit. We implement our methods in the python package \texttt{mosaicperm}.
\end{abstract}

\section{Introduction}\label{sec::intro}

\subsection{Motivation and problem statement}\label{subsec::motivation}

Fundamental factor models are among the most common statistical tools used to manage risk in economics and finance \citep{rosenberg1974, kahn1994multiplefactor, dacheng2022annualreview}. Indeed, analysts routinely use these models to model the correlations among asset returns, allowing one to estimate the probability that many assets in a portfolio will simultaneously lose value.
Yet as conditions change, analysts must assess whether established models remain reliable. As an illustrative example, this paper analyzes the BlackRock Fundamental Equity Risk (BFRE) model, one of many commercially available risk models used in industry. In particular, Section \ref{sec::realdata} asks whether the BFRE model adequately explains correlations among US stock returns three months after the COVID-19 pandemic began. Correctly answering such questions is essential: on the one hand, needlessly rebuilding an established model may waste resources and ultimately increase risk, but on the other hand, it is important to quickly detect inadequacies in existing models.

This article develops statistical methods to test the goodness-of-fit of existing fundamental factor models. Formally, at times $t=1,\dots,T$, suppose we observe returns $Y_t \in \R^p$ for $p$ assets which we believe follow the model
\begin{equation}\label{eq::factormodel}
    Y_t = L_t X_t + \epsilon_t,
\end{equation}
where $L_t$, $X_t$ and $\epsilon_t$ are defined below:
\begin{itemize}[itemsep=0.5pt, topsep=1pt, leftmargin=*]
    \item $X_t \in \R^k$ denotes the returns of $k \ll p$ underlying factors which drive correlation among the assets. 
    We assume the factor returns $X_t$ are not observed.

    \item $L_t \in \R^{p \times k}$ are factor ``loadings" or exposures, i.e., $[L_t]_{j\ell}$ measures the exposure of the $j$th asset to the $\ell$th factor at time $t$. We treat $L_t$ as a deterministic matrix that is known at time $t$ (see below).

    \item $\epsilon_t \in \R^p$ denotes the idiosyncratic returns of the $p$ assets which cannot be explained by the factors. We also refer to $\epsilon_t$ as the ``residuals." 
\end{itemize}

This paper analyzes \textit{fundamental risk models} like the BFRE model, where the exposures $L_t$ are based on market fundamentals such as industry membership and accounting data. For example, $[L_t]_{j\ell} \in \{0,1\}$ might indicate whether stock $j$ is in the $\ell$th industry. Unlike some factor models commonly used in, e.g., academic finance or psychology, this means that the exposures $L_t$ are known at time $t$, although the factor returns $X_t$ are not observed and must be estimated, typically using cross-sectional regressions. Although this framework is perhaps less common in the academic literature, it is widely used in industry. Indeed, a recent review \citep{dacheng2022annualreview} noted that fundamental factor models provide ``arguably the most prevalent [factor model] framework for practitioners."

Naturally, other risk models exist, including (i) \textit{macrofactor risk models}, where $X_t$ denotes observed macroeconomic data and $L_t$ is unknown, and (ii) \textit{purely statistical models}, where both $X_t$ and $L_t$ are estimated. Such models are beyond the scope of this paper.\footnote{It is possible to extend the methods in this paper to the case where $X_t$ is observed and $L_t$ is not. However, it requires rather different statistical techniques, so we defer this extension to a companion paper \citep{companionfrt2023}.}

To test if (\ref{eq::factormodel}) accurately models correlations among assets, let $\epsilon_{\cdot,j} \defeq (\epsilon_{1,j}, \dots, \epsilon_{T,j}) \in \R^T$ denote the time series of residuals for the $j$th asset. We will test the null that the residual processes are independent across assets:
\begin{equation}\label{eq::nullhypothesis}
    \mcH_0 : \epsilon_{\cdot,1}, \epsilon_{\cdot,2}, \dots, \epsilon_{\cdot,p} \in \R^T \text{ are jointly independent.} 
\end{equation}

As stated, $\mcH_0$ allows the residuals to be autocorrelated, nonstationary, and heterogeneous. Indeed, $\mcH_0$ allows temporal dependence among the residuals of the $j$th asset, although it requires all residuals of the $j$th asset to be independent of all other residuals. And so far, $\mcH_0$ does not require any stationarity conditions, allowing the volatility of each asset's residuals to change over time (although we will assume a ``local exchangeability" condition in Section \ref{subsec::defaulttiling} which restricts the nonstationarity pattern). Lastly, $\mcH_0$ allows heterogeneity in the sense that the joint law of the residuals for asset $j$ may be arbitrarily different from the corresponding law for asset $j' \ne j$. Since autocorrelation, nonstationarity, and heterogeneity are important features of real financial data, our work attempts to avoid making unrealistic assumptions about these phenomena. 

We emphasize that we seek to test whether $\mcH_0$ holds for a \textit{fixed} choice of exposures $L_t$. In contrast, many previous works test whether $\mcH_0$ holds for some \textit{unknown} choice of $L_t \in \R^{p \times k}$ or estimate the number of factors $k$ (see Section \ref{subsec::literature}). These latter problems have other applications, but they do not accomplish our goal, which is to test if a pre-existing risk model continues to fit the data. Indeed, financial firms routinely publish exposure matrices $L_t$ for commercial risk models, including MSCI Barra models and the BFRE model \citep{rosenberg1976barra, bender2012barra}.\footnote{Naturally, our methods also apply if one selects exposures using historical data and tests the selected model's goodness-of-fit on fresh data.} Since the exposures are typically constructed using observable quantities such as industry membership, it makes sense to treat them as a known quantity instead of an imprecise estimate. Yet given exposures $\{L_t\}_{t=1}^{T}$ from one of these models, how should financial firms interpret evidence against the null and quantify uncertainty? This is the type of question our work aims to answer---please see Section \ref{subsec::bfre_motivation} for a concrete motivating example.

We argue that a good test of $\mcH_0$ should rigorously control false positives, i.e., it should reject $\mcH_0$ with probability at most $\alpha$ whenever $\mcH_0$ actually holds. This is important for several reasons. First, it is important for risk management: in times of volatility, needlessly doubting a well-fitting risk model could be just as harmful as relying on a misspecified one. Second, large financial firms may constantly stress-test their risk models. Without rigorous false positive control, they may discard and rebuild many well-fitting models for no reason, consuming a great deal of resources and possibly reducing model quality in the long run. Lastly, rigorous hypothesis tests may be helpful during the process of constructing the exposures $L_t$, because they quantify evidence against different candidate models. 

These arguments are not new, and the problem of testing the goodness-of-fit of a factor model dates back to the origins of the field of statistics \citep{spearman1904,roy1953, bartlett1954chi2,boxanderson1955, lawley1956,horn1965}. These seminal works established that if (i) the idiosyncratic returns $\epsilon_t$ are Gaussian or (ii) the number of assets $p$ is held constant as the number of timepoints $T$ diverges, one can perform (asymptotically) valid hypothesis tests using the generalized likelihood ratio (GLR) test (see \cite{anderson2009}). Alternatively, other classical approaches apply the bootstrap \citep{efron1979bootstrap} or the block bootstrap \citep[e.g.,][]{kunsch1989block_bootstrap, romano2006block_bootstrap}.

However, as noted by the modern literature (see Section \ref{subsec::literature}), these classical techniques may not be suited to modern datasets. First, modern applications are typically high-dimensional, meaning that the number of assets $p$ is comparable to or much larger than $T$. 
Indeed, to quickly detect violations of $\mcH_0$, we might analyze datasets with $p \ge 2000$ assets and $T \approx 50$ datapoints, or subsectors with $p \approx 150$ assets and $T \approx 50$ observations. In these settings, classical theory for likelihood ratios and bootstrap methods will generally be inaccurate \citep[e.g.,][]{bai2003, elk2018bootstrap, sur2019pnas}. Even high-dimensional asymptotics \citep[e.g.,][]{bai2003, bai2016review} may be inaccurate when $T$ is small and the ratio $k/p$ is not negligible (see Section \ref{subsec::literature}). Second, parametric assumptions are not appropriate, since real data may exhibit features that are not captured by the model, such as heavy tails or heteroskedasticity. Lastly, even if we could apply classical likelihood-based theory, we might prefer to use regularization or other machine learning techniques to increase power. Thus, in this paper, we ask: can we develop finite-sample valid tests of $\mcH_0$ under no parametric assumptions? Additionally, can we leverage prior information and machine learning techniques to increase power while controlling false positives?

\subsection{A motivating application to the BlackRock Fundamental Equity Risk model}\label{subsec::bfre_motivation}

As a motivating example, we analyze the BlackRock Fundamental Equity Risk (BFRE) model, a factor model that publishes weekly exposure matrices $L_t$ for $p \ge 2000$ US stocks and $k \approx 65$ factors. We ask: three months after the COVID-19 pandemic began, does the BFRE model adequately explain the correlations among US stock returns?

One reasonable way to answer this question might be to track, for each asset, the maximum absolute correlation between its residuals and those of another asset. Formally, we:

\begin{enumerate}[itemsep=0.1pt, topsep=0pt, leftmargin=*]
\item Run cross-sectional ordinary least squares (OLS) regressions to estimate the residuals. I.e., for the standard OLS projection matrix $H_t\ols \defeq I_p - L_t (L_t^\top L_t)^{-1} L_t^\top$, define:
\begin{equation}\label{eq::hateps_ols_def}
\hat \epsilon_t\ols \defeq H_t\ols Y_t = H_t\ols (L_t X_t + \epsilon_t) = H_t\ols \epsilon_t.
\end{equation}
\item Compute the empirical correlation matrix $\hat C \in \R^{p \times p}$ of the last $350$ estimated residuals $\{\hat \epsilon_s\ols\}_{s=t-350}^{t}.$ (The choice of window size is somewhat arbitrary, but Appendix \ref{appendix::windowsize} obtains similar results using many different window sizes.)

\item For each asset $j \in [p]$, define $\hmc_j \defeq  \max_{j' \ne j} |\hat C_{j, j'}|$ as the maximum estimated absolute correlation between asset $j$ and another asset.
\item Finally, let $S_t\ols = \frac{1}{p} \sum_{j=1}^p \hmc_j$ denote the \textit{mean maximum (absolute) correlation} (MMC) over all assets at time $t$, which we use as an aggregate measure of model fit. Figure \ref{fig::motivation} plots $S_t\ols$ biweekly for three sectors using data from 2017 through 2023.
\end{enumerate}

\begin{figure}[!h]
    \includegraphics[width=\linewidth]{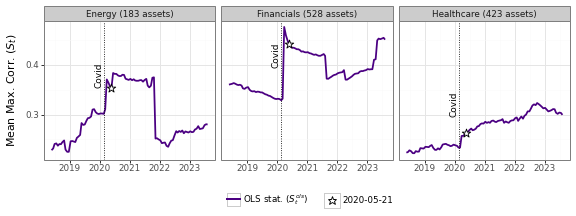}
    \caption{For three industries, this figure plots biweekly values of $S_t\ols = \frac{1}{p} \sum_{j=1}^p \hmc_j$, where $\hmc_j = \max_{j' \ne j} |\hat C_{j,j'}|$ and $\hat C$ is the empirical correlation matrix of the estimated idiosyncratic returns $\{\hat \epsilon_s\ols\}_{s-350}^t$ from the last $350$ days. Interpreting this plot is challenging because it is not obvious what one should expect to see, even when the factor model fits perfectly. In general, $S_t\ols$ is neither mean-zero nor stationary under $\mcH_0$. Indeed, the large jumps in February 2020 coincide with large increases in the variance of $\epsilon_t$.}
    \label{fig::motivation}
\end{figure}

Producing plots like Figure \ref{fig::motivation} is easy---however, \textit{interpreting} these plots is hard, because it is not obvious what types of fluctuations we would see under the null. For example, even when $\epsilon_t$ has independent components, $\cov(\hat\epsilon_t\ols) = H_t\ols \cov(\epsilon_t) H_t\ols$ is not diagonal. Furthermore, the law of $S_t\ols$ should change over time, because (a) the projection matrices $H_t\ols$ change over time due to the changing exposures and (b) the law of the residuals drifts over time---e.g., after COVID, the residuals become more heavy-tailed.

These facts make it hard to interpret Figure \ref{fig::motivation}. Certainly, in all three sectors, the statistic $S_t$ jumps after February 2020---but are these jumps consistent with the model? 
For example, on \stardate, we observe absolute correlations of $0.35, 0.44$, and $0.26$ in the energy, financial, and healthcare sectors. Should BlackRock devote time and resources to improve the model based on these plots? Our work aims to provide statistical tools to answer these questions, not only for this particular test statistic but for many goodness-of-fit measures. 

\subsection{Contribution}

Our paper introduces an exact and nonparametric permutation test of $\mcH_0$. The key idea is to introduce a new estimator $\hat\epsilon \in \R^{T \times p}$ of the residuals that exactly preserves some of the independence properties of the true residuals $\epsilon$. To construct $\hat\epsilon$, we split the data matrix $\bY$ into rectangular tiles (along both axes) and separately estimate the residuals in each tile, yielding a residual estimate $\hat \epsilon \in \R^{T \times p}$. Given a test statistic $S(\hat \epsilon)$ quantifying the correlations among the columns of $\hat \epsilon$, we can compute a significance threshold for $S(\hat \epsilon)$ by permuting the observations in each tile of $\hat \epsilon$. The idea is illustrated in Figure \ref{fig::intro_figure}. We call this a ``mosaic permutation test" because the separation of the data into tiles is reminiscent of a mosaic.

\begin{figure}[!h]
    \centering
    \includegraphics[width=\linewidth]{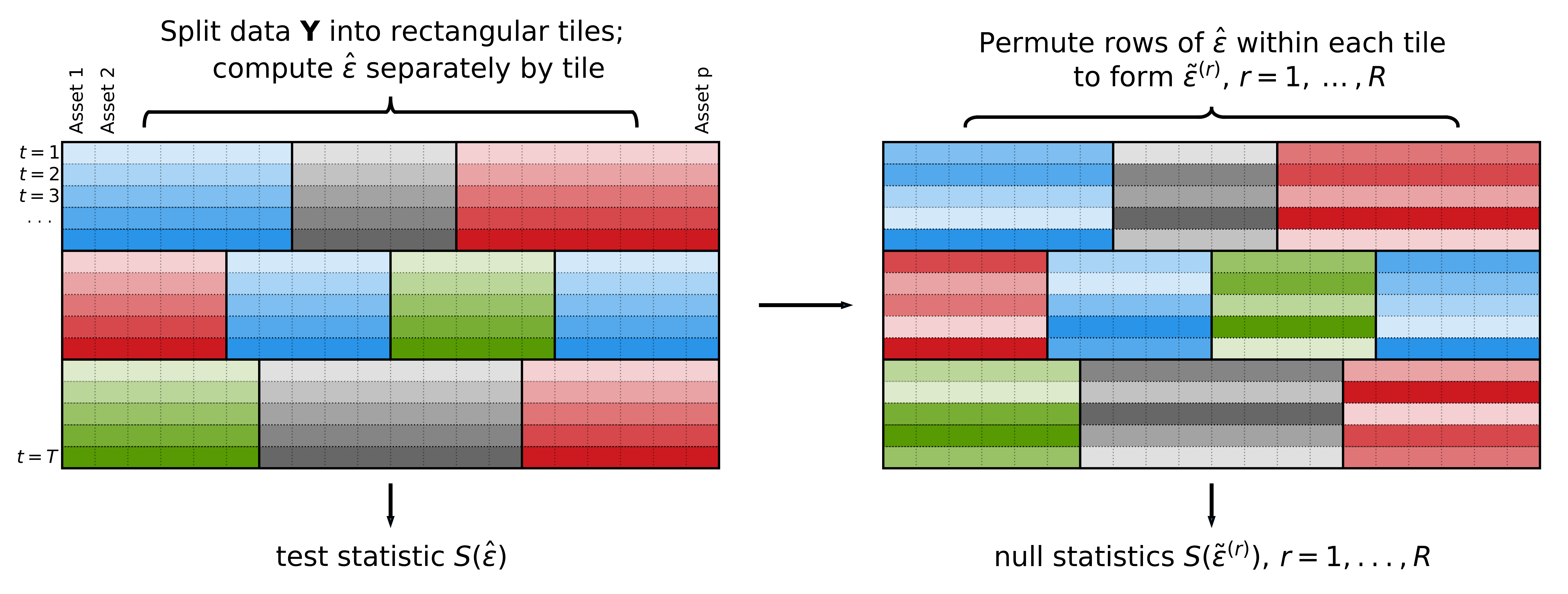}
    \caption{This figure summarizes the main methodology of the mosaic permutation test. Above, rows represent different observations $t=1,\dots,T$, columns represent different assets, and in the second matrix, the shadings are permuted within each rectangle to illustrate the permutations within tiles.}
    \label{fig::intro_figure}
\end{figure}

We now highlight a few key properties of the mosaic permutation test.

\textbf{1. Exact and nonparametric false positive control.} The test yields an exact p-value in finite samples under only the assumption that the residuals for each asset are locally exchangeable (defined in Section \ref{sec::methodology}). In particular, we make no assumptions about the marginal distributions of $\{\epsilon_t\}_{t=1}^{T}$ and $\{X_t\}_{t=1}^{T}$, allowing them to be arbitrarily heavy-tailed. Furthermore, to allow for changing market conditions, our results allow the factor returns $\{X_t\}_{t=1}^{T}$ to be arbitrarily non-stationary and the residuals $\{\epsilon_t\}_{t=1}^T$ to be non-stationary across tiles. For instance, in our empirical application, we require that the residuals are stationary within each week, but their distributions can change arbitrarily between weeks.

To illustrate this contribution, we conduct semisynthetic simulations using the exposure matrix $L_t$ from the BFRE model for energy stocks on \stardate. For simplicity, we generate data $\bY$ from Eq. \ref{eq::factormodel} after sampling the residuals and factor returns as i.i.d. standard Gaussians, with $T=350$ observations. We use the test statistic from Section \ref{subsec::bfre_motivation}. Figure \ref{fig::bsbad} shows that in this simple Gaussian setting, naive bootstrap and permutation testing methods yield essentially a $100\%$ false positive rate (we review intuition for this result in Section \ref{sec::failure}). 
In contrast, the mosaic permutation test has provable validity in finite samples.

\begin{figure}
    \includegraphics[width=\linewidth]{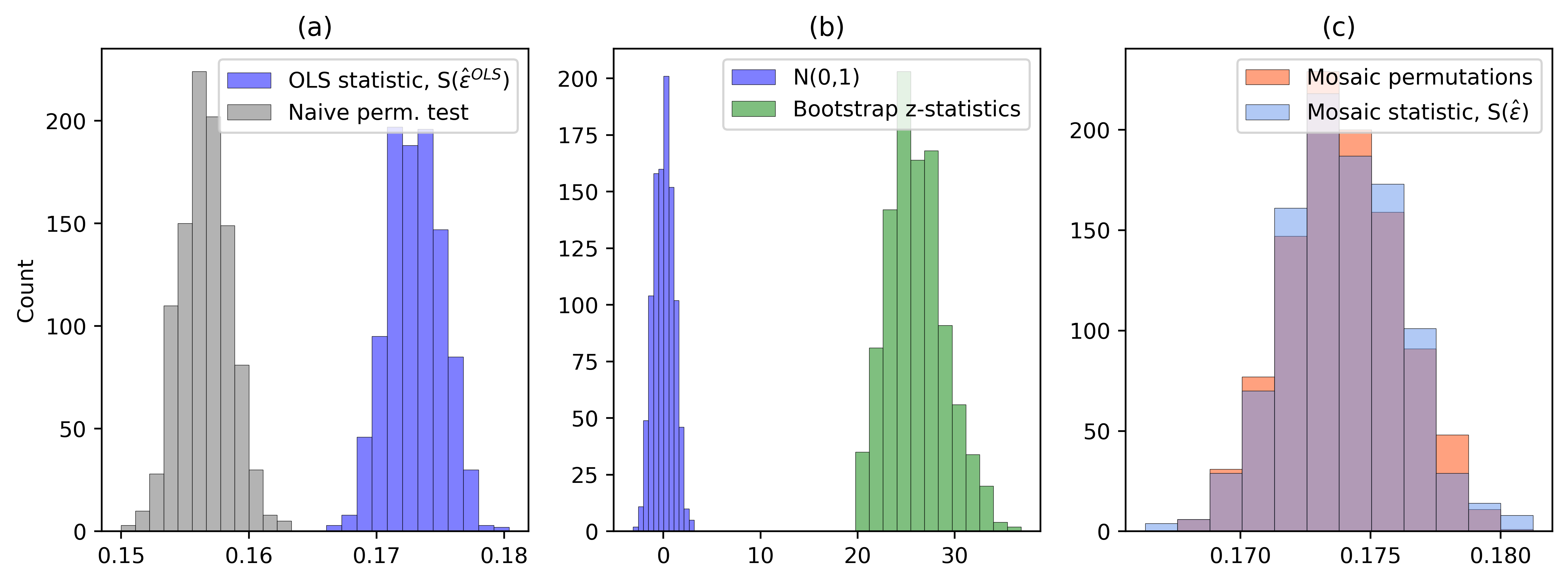}
    \caption{Semisynthetic simulation with $X_{tk}, \epsilon_{tj} \iid \mcN(0,1)$, and the exposures are taken from the BFRE model for the energy sector on \stardate. Note $T=350$, $p=183$, $k=18$, and we use the test statistic from Figure \ref{fig::motivation}. Panel 2(a) shows that a naive residual permutation test (discussed in Section \ref{subsec::naive_perm_test}) inaccurately simulates the null distribution of the test statistic $S(\hat\epsilon\ols)$---in fact, the true null distribution and the estimated one do not overlap. Panel 2(b) shows that naive bootstrap Z-statistics (discussed in Section \ref{subsec::naive_bootstrap}) are not approximately mean zero, nor do they have unit variance. All p-values based on these two naive methods are numerically indistinguishable from zero, leading to an empirical false positive rate of $100\%$. In contrast, the mosaic permutation test uses a different ``mosaic" estimator of the residuals $\hat\epsilon$. Using $\hat\epsilon$ in place of $\hat\epsilon\ols$ allows us to use a permutation method to accurately simulate the law of $S(\hat\epsilon)$ under the null, as shown in Panel 2(c)---see Section \ref{sec::methodology} for details.}\label{fig::bsbad}
\end{figure}

\textbf{2. Power and flexibility.} Our method permits the use of a wide set of test statistics while retaining provable false positive control. For example, it allows analysts to use regularized estimates of the covariance matrix of the residuals $\epsilon$, for example, via a graphical lasso \citep{glasso2007}, and it also permits the use of cross-validation to choose the regularization strength. The only restriction is that the test statistic must be a function of the mosaic estimator $\hat\epsilon$ of the residuals instead of a function of (e.g.) a conventional OLS estimator $\hat\epsilon\ols$. This is the price our method pays to quantify uncertainty.

That said, in semisynthetic simulations based on US stock data, we find that the mosaic test does not lose much power compared to an ``oracle" test based on $\hat\epsilon\ols$ (see Figure \ref{fig::mainsims}).

\textbf{3. Interpretable test statistics.} Our primary methodological contribution is to develop finite-sample tests of $\mcH_0$. However, once one has rejected $\mcH_0$, it may be of interest to learn how to \textit{improve} the factor model. With this motivation, our empirical analysis (Section \ref{subsec::improvement}) introduces practical test statistics which (i) adaptively estimate the sparsity of any missing factor exposures and (ii) are designed to diagnose when an estimate of missing factor exposures truly improves the model fit. That said, this is a secondary focus of our work---our primary focus is to detect model violations.

\textbf{Empirical application}: To illustrate our method, we test the goodness-of-fit of the BlackRock Fundamental Equity Risk (BFRE) model, from 2017 through 2023. Our analysis distinguishes between \textit{persistent} factors, which have explanatory power over long time periods, and \textit{transient} (non-persistent) factors. (Commercial models typically focus on including persistent factors but not necessarily transient ones.) We report three overall findings:
\begin{enumerate}[noitemsep, topsep=0pt, leftmargin=*]
    \item The BFRE model appears to explain the most significant correlations among asset returns, and in most sectors, we are unable to substantially improve the model fit.
    \item However, we find evidence of statistically significant unexplained correlations among (i) real estate stocks and (ii) healthcare stocks post-COVID. The correlations among healthcare stocks appear to be \textit{transient}, as incorporating them does not consistently improve the model. However, we show that adding a factor to account for extra correlations among real estate stocks \textit{persistently} improves the model fit.
    \item In contrast, removing existing factors from the model leads to much stronger evidence against the null. 
\end{enumerate}

We present our findings in detail in Section \ref{sec::realdata}. However, for illustration, Figure \ref{fig::motivation_soln} shows the results after applying the mosaic permutation test to the analysis in Figure \ref{fig::motivation}. As of \stardate, we do not find statistically significant evidence against $\mcH_0$ in the energy sector, but we do find evidence of unexplained correlations in the financial and healthcare sectors. This result is not obvious before one tests for statistical significance, as Figure \ref{fig::motivation_soln} confirms that the significance threshold for $S_t$ is not constant over time. 

\begin{figure}[!h]
    \includegraphics[width=\linewidth]{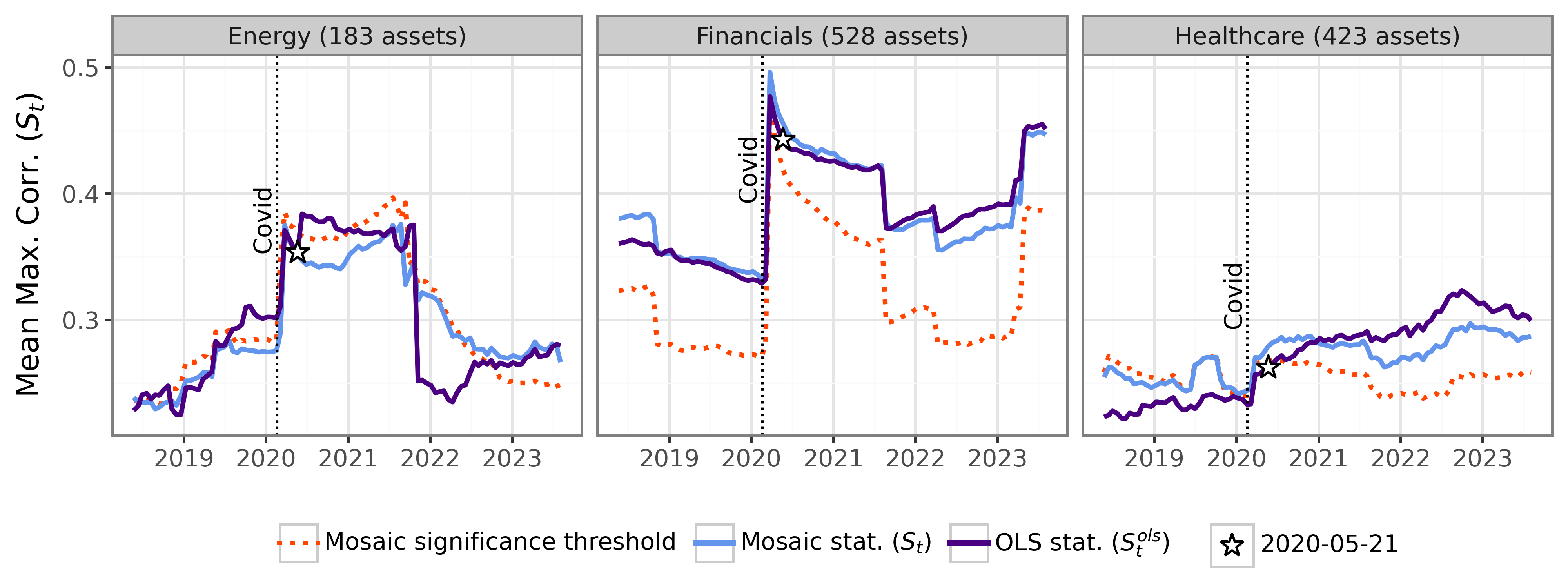}
    \caption{This plot is identical to Figure \ref{fig::motivation} with two additions. First, we compute a mosaic version $(S_t)$ of the OLS test statistic $(S_t\ols)$ by replacing the OLS residual estimates $\hat\epsilon_t\ols$ with mosaic estimates $\hat\epsilon_t$---see Sections \ref{sec::methodology}-\ref{sec::realdata} for details. Second, the dotted orange line displays the mosaic significance threshold, so $S_t$ is statistically significant at time $t$ if the blue line lies above the orange line. To better visualize the correlation between $S_t\ols$ and $S_t$, in this plot only, we shift the test statistic $S_t$ up by a small constant shift $(0.06, 0.05$ and $0.04$ in Energy, Financials, and Healthcare, respectively) so that average plotted values of $S_t$ and $S_t\ols$ are equal. This constant shift provably does not affect the mosaic p-value since the significance threshold shifts upwards by the same amount.
}
    \label{fig::motivation_soln}
\end{figure}

\subsection{Related literature}\label{subsec::literature}

Beyond the classical methods mentioned in Section \ref{subsec::motivation}, our work contributes to a wide statistical literature on inference for factor models, including inference in the high-dimensional setting where $p$ grows with $T$ \citep[e.g.,][]{bai2003, onatski2012, bai2016review},
inference on the number of factors $k$ \citep[e.g.,][]{onatski2009econometrica, onatski2010, alessi2010nfactors, owen2016bicross, dobriban2018deterministic_pa, dobriban2020_pa_theory}, tests for changepoints in the factor loadings \citep[e.g.,][]{breitung2011changepoint, bai2022likelihood_changepoint}, tests for whether observed proxies of the factor returns $X_t$ are good proxies \citep{bai2006evalproxies}, bootstrap methods to debias OLS estimates of the exposures \citep{goncalves2020factorbootstrap}, covariance matrix estimation \citep[e.g.,][]{fan2008covest, fan2011_cov_est_apprx_factor}, estimation and inference for dynamic factor models \citep[e.g.,][]{bates2013dynamic_factor_models, stock_watson_2016_review} and more---see \cite{bai2016review} for a review. Many of these techniques leverage key results from random matrix theory \citep{johnstone2001pca, paul2007, silversteinbai2010book}, sometimes in combination with permutation-based methods \citep{buja1992pa_permute}. 

Factor models are also widely discussed in the asset pricing literature (see \cite{dacheng2022annualreview} for review), including methods for explaining excess returns \citep{fama_french_1992, fama_french_2008, welch_goyal_2007, lewellen2015, freyberger_nonparametric_asset_pricing2020, gu2020_ml4assetpricing}, estimating factors and exposures \citep{chamberlain_rothschild_1983, connor_korajczyk_1986, kelly2019_instrumented_pca, lettau2020_theory}, estimating notions of factor strength \citep{bailey2021factorstrength}, estimating risk premia and stochastic discount factors \citep{pesaran2019riskpremia, kozak2020shrink_cross_section,  giglio2021testassets, anatolyev2022}, testing asset pricing models \citep{gibbons1989, campbell2006, fan2015_sparsepowerenhancement, pukthuanthong2018, jegadeesh2019, raponi2019largeN_fixedT}, comparing asset pricing models \citep{hansen1997_HJstatistic, kan2008, kan2013, harvey2015multipletesting}, and more. We note that this paper focuses on testing \textit{risk} models as opposed to asset-pricing models (see below); thus our methodology and empirical results are complementary to those in this literature.

Indeed, our work differs from the existing literature in three respects.
\begin{itemize}[leftmargin=*, itemsep=0.5pt, topsep=0pt]
    \item First, we solve a different problem. Motivated by real financial applications, we seek to test the goodness-of-fit of a risk model with known exposures $\{L_t\}_{t \in [T]}$. In contrast, most works in the statistical literature treat the exposures as unknown nuisance parameters, and we are not aware of existing works that explicitly test $\mcH_0$ for a known choice of exposures. 
    Furthermore, the asset pricing literature focuses on testing the economic assumptions underlying asset pricing models \citep[e.g.,][]{gibbons1989, jegadeesh2019} or testing whether a given factor partially explains the cross-section of expected asset returns \citep[e.g.,][]{harvey2015multipletesting}. In contrast, motivated by our focus on risk models, we test whether the residuals $\epsilon_t$ have independent components.
        
    \item  Second, even for the problems they solve, existing works only provide asymptotic control of the false positive rate. Such results typically require technical assumptions controlling the heavy-tailedness of the data, the stationarity of the factors and residuals, the strength of the factors, and the rates at which $T$ and $p$ diverge (see, e.g., \cite{bai2016review, raponi2019largeN_fixedT} for longer reviews). It is not clear that these methods can satisfactorily control false positives (i) if the technical assumptions are violated or (ii) if $T$ and/or $p$ are small (e.g., if $T \le 50$ and the number of factors $k$ is not negligible compared to the number of assets $p$). In contrast, our method exactly controls false positives in finite samples under a local exchangeability condition (see Section \ref{sec::methodology}).
    
    \item  Lastly, the vast majority of existing works require the analyst to use specific test statistics when performing hypothesis tests, such as likelihood ratios or the eigenvalues of the empirical covariance matrix \citep[e.g.,][]{bai2003, onatski2009econometrica, alessi2010nfactors, breitung2011changepoint, bai2022likelihood_changepoint}. In contrast, our work allows the use of regularization and black-box machine learning techniques to quantify evidence against the null.
\end{itemize}

We note also that a recent literature has established the asymptotic consistency of various bootstrap methods when approximating the law of maxima of sums of high-dimensional random vectors \cite[e.g.,][]{chern2013mbs, chern2017cltbs, deng2020hdbs, lopes2020hdbs, chern2022improvedclt, cherno2023hdbsreview}. Yet in our setting, we only observe estimated residuals $\{\hat\epsilon_t\}_{t=1}^T$ which may have a significantly different covariance matrix than the true residuals---indeed, if the number of factors $k$ is not negligible compared to the number of assets $p$, this difference is nonvanishing. This introduces additional bias, and it is not clear how to use the bootstrap to correct this (see Section \ref{subsec::naive_bootstrap}). Furthermore, our work provides methods to perform inference based on a wide variety of test statistics beyond the maxima of sums.

Finally, our work contributes to a growing literature on exact finite-sample permutation tests for linear models \citep{lei2020cpt, wen2022ResidualPT, dhaultfoeuille2022ARP, guan2023PALMRT}. However, these tests are not designed to apply to factor models, and they would only apply if all idiosyncratic returns have the same distribution, which is not realistic, since (e.g.) the variance of the idiosyncratic returns typically varies substantially across assets. In contrast, our theory allows the distribution of the idiosyncratic returns to vary arbitrarily across assets.

\subsection{Notation}

For $n \in \N$, define $[n] \defeq \{1,\dots,n\}$. For any $A \in \R^{n_1 \times n_2}$, $A_{i} \in \R^{n_2}$ and $A_{\cdot,j} \in \R^{n_1}$ denote the $i$th row and $j$th column of $A$, respectively. For subsets $I \subset [n_1], J \subset [n_2]$, $A_I \in \R^{|I| \times n_2}$ denotes the submatrix formed by the rows in $I$, $A_{\cdot,J} \in \R^{n_1 \times |J|}$ denotes the submatrix formed by the columns in $J$, and $A_{I,J} \in \R^{|I| \times |J|}$ denotes the submatrix formed by the rows in $I$ and the columns in $J$. $\epsilon \defeq \begin{bmatrix} \epsilon_1 & \dots & \epsilon_T \end{bmatrix}^{\top} \in \R^{T \times p}$
denotes the matrix of residuals and $\bY \in \R^{T \times p}$ denotes the returns. Thus, $\epsilon_t \in \R^p$ denotes the vector of all $p$ assets' residuals at time $t$, whereas $\epsilon_{\cdot,j} \in \R^{T}$ denotes the time series of residuals for asset $j$. We let $\hat\epsilon\ols \in \R^{T \times p}$ denote estimates of the residuals formed using cross-sectional OLS (see Eq. \ref{eq::hateps_ols_def}). $\hat\epsilon \in \R^{T \times p}$ denotes the mosaic estimates of the residuals, introduced in Section \ref{sec::methodology}.

\section{Performance of default bootstrap and permutation methods}\label{sec::failure}

To review from Section \ref{sec::intro}, the problem statement is to test the following factor model:
\begin{align}
    Y_t = L_t X_t + \epsilon_t & \text{ for } t = 1, \dots, T,
\end{align}
for outcomes $Y_t \in \R^p$, fixed and known exposures $L_t \in \R^{p \times k}$, unobserved factor returns $X_t \in \R^k$ and residuals $\epsilon_t \in \R^p$. We seek to test the null hypothesis $\mcH_0$ that the time series of residuals for each asset are independent:
\begin{equation}
    \mcH_0 : \epsilon_{\cdot,1}, \epsilon_{\cdot,2}, \dots, \epsilon_{\cdot,p} \in \R^T \text{ are jointly independent.}
\end{equation}
Sections \ref{subsec::naive_perm_test}-\ref{subsec::naive_bootstrap} now explain why naive permutation and bootstrap tests can yield false positive rates of up to $100\%$, as in Figure \ref{fig::bsbad}. The main challenge is that the estimated OLS residuals $\hat\epsilon\ols$ do not satisfy the same independence properties as the true residuals.%, and (ii) our problem setting is too high-dimensional for the bootstrap to perform well \citep{bickelfreedman1983bootstrap, elk2018bootstrap}.

\subsection{Naive residual permutation tests are invalid}\label{subsec::naive_perm_test}

For simplicity, we assume for this section that the vectors of residuals $\epsilon_1, \dots, \epsilon_T \iid P_{\epsilon}$ are i.i.d. This assumption plus $\mcH_0$ together imply that separately permuting the residuals of each asset does not change the joint law of all of the residuals:

\begin{equation}\label{eq::epspermv2}
\epsilon \defeq 
\begin{bmatrix}
    \epsilon_{1,1} & \epsilon_{1,2} & \dots & \epsilon_{1,p} \\
    \epsilon_{2,1} & \epsilon_{2,2} & \dots & \epsilon_{2,p} \\
    \epsilon_{3,1} & \epsilon_{3,2} & \dots & \epsilon_{3,p} \\
    \vdots & \vdots & \dots & \vdots \\
    \epsilon_{T,1} & \epsilon_{T,2} & \dots & \epsilon_{T,p}
\end{bmatrix}
\disteq 
\begin{bmatrix}
    \epsilon_{\cblue{\pi_1}(1),1} & \epsilon_{\cred{\pi_2}(1),2} & \dots & \epsilon_{\cgreen{\pi_p}(1),p} \\
    \epsilon_{\cblue{\pi_1}(2),1} & \epsilon_{\cred{\pi_2}(2),2} & \dots & \epsilon_{\cgreen{\pi_p}(2),p} \\
   \epsilon_{\cblue{\pi_1}(3),1} & \epsilon_{\cred{\pi_2}(3),2} & \dots & \epsilon_{\cgreen{\pi_p}(3),p} \\
    \vdots & \vdots & \dots & \vdots \\
   \epsilon_{\cblue{\pi_1}(T),1} & \epsilon_{\cred{\pi_2}(T),2} & \dots & \epsilon_{\cgreen{\pi_p}(T),p}
\end{bmatrix},
\end{equation}
where above, $\pi_1, \dots, \pi_p : [T] \to [T]$ are arbitrary permutations applied to the $p$ columns of $\epsilon$. Thus, if we observed $\epsilon$, we could easily design a permutation test of $\mcH_0$ as follows.
\begin{enumerate}[leftmargin=*, topsep=0.5pt, itemsep=0.5pt]
    \item Permute each of the columns of $\epsilon$ uniformly at random, and repeat this $R$ times, yielding permuted matrices $\epsilon^{(1)}, \dots, \epsilon^{(R)} \in \R^{T \times p}$.
    \item Let $S(\epsilon)$ be any test statistic, such as the maximum empirical correlation among the residuals. Compute a p-value by comparing the value of $S(\epsilon)$ to $S(\epsilon^{(1)}), \dots, S(\epsilon^{(R)})$:
    \begin{equation}\label{eq::epspvalv2}
        \pval \defeq \frac{1 + \sum_{r=1}^R \I(S(\epsilon) \le S(\epsilon^{(r)}))}{R+1},
    \end{equation}
    where Equation \ref{eq::epspermv2} guarantees that this is a finite-sample p-value testing $\mcH_0$.
\end{enumerate}

Although we do not observe the residuals $\epsilon$, a ``naive residual permutation test" would simply plug in the OLS estimate $\hat\epsilon\ols \in \R^{T \times p}$ in place of $\epsilon$ (see Eq. \ref{eq::hateps_ols_def}). Unfortunately, this strategy will not work: while the coordinates of $\epsilon_t$ are independent under $\mcH_0$, the coordinates of $\hat \epsilon\ols$ are not. Indeed, $\cov(\hat \epsilon_t\ols) = H_t\ols \cov(\epsilon_t) H_t\ols$ has a reduced rank of at most $p-k$. As a result, naively replacing $\epsilon$ with $\hat\epsilon\ols$ will violate Eq. \ref{eq::epspermv2}---even under the null, the columns of $\hat\epsilon\ols$ will look more correlated than the permuted version of $\hat\epsilon\ols$. Indeed, Figure \ref{fig::bsbad}(a) uses the real BFRE model exposures to show that this ``naive permutation test" may cause an unacceptably high false positive rate. 

\subsection{Naive bootstrap methods are invalid}\label{subsec::naive_bootstrap}

Another way to test $\mcH_0$ would be to use the nonparametric bootstrap to compute a Z-statistic based on $S(\hat\epsilon\ols)$. This strategy does not adjust the estimated residuals to force them to satisfy the null---rather, it reframes the hypothesis testing problem as an estimation problem. Namely, suppose that $S(\hat\epsilon\ols)$ is an estimate of some parameter $\theta$ which equals zero under $\mcH_0$. For example, in Section \ref{subsec::bfre_motivation}, $S_t\ols$ is a (biased) estimate of the true mean absolute maximum correlation of the residuals, which is zero under $\mcH_0$. We might hope to use the bootstrap to debias and standardize $S(\hat\epsilon)$. While there are many ways to apply the bootstrap, perhaps the simplest (and most naive) is as follows:

\begin{enumerate}[leftmargin=*, noitemsep, topsep=0pt]
    \item Resample $T$ rows from $\hat\epsilon\ols$ uniformly at random and with replacement.\footnote{Since the test statistic only depends on $\bY$ through $\hat\epsilon\ols$, the residual bootstrap is in this case identical to the pairs bootstrap, which resamples pairs of exposures and returns $\{(L_t, Y_t)\}_{t=1}^T$.}
    \item Repeat this $B$ times, yielding $B$ bootstrapped residual matrices $\hat\epsilon^{\mathrm{ols},(1)}, \dots, \hat\epsilon^{\mathrm{ols},(B)} \in \R^{T \times p}$.
    \item Compute a bootstrap bias estimate for $S(\hat\epsilon\ols)$ as well as a $Z$-statistic which is intended to have zero mean and unit variance under the null:
    \begin{equation}
        \widehat{\mathrm{Bias}} = \frac{1}{B} \sum_{b=1}^B S(\hat\epsilon^{\mathrm{ols},(b)}) - \theta_{\mathrm{BS}} \,\,\,\,\, \text{ and }  \,\,\,\,\, Z_{\mathrm{BS}} = \frac{S(\hat\epsilon\ols) - \widehat{\mathrm{Bias}}}{\sqrt{\widehat{\var}(\{S(\hat\epsilon^{\mathrm{ols},(b)})\}_{b=1}^B)}},
    \end{equation}
    where $\theta_{\mathrm{BS}}$ is the value of the parameter $\theta$ calculated for the bootstrap empirical distribution.
\end{enumerate}

Unfortunately, inference based on the procedure above can be highly misleading, as the bootstrap distribution for $\{\hat\epsilon_t\}_{t=1}^T$ may not accurately approximate the law of $\{\epsilon_t\}_{t=1}^T$, for at least two reasons. First, the law of $\hat\epsilon_t$ is not the same as the law of $\epsilon_t$, since under the null, $\cov(\epsilon_t)$ is diagonal but $\cov(\hat\epsilon_t)$ is not, and this difference is nonvanishing when $k$ is not negligible compared to $p$. Second, the estimated residuals $\hat\epsilon\ols_t \in \R^p$ are ``high-dimensional" vectors in the sense that the number of assets $p$ is usually not negligible compared to the number of observations $T$. Thus, the bootstrap distribution of $S(\hat\epsilon\ols)$ may not accurately approximate the true law of $S(\hat\epsilon\ols)$ \citep{bickelfreedman1983bootstrap, elk2018bootstrap}. Indeed, even in simpler contexts, such high-dimensional approximation results are only known for certain classes of test statistics \citep{cherno2023hdbsreview}. Thus, the bias and variance estimates can be highly inaccurate, since they are based on a bootstrap distribution that differs substantially from the true data-generating process. And despite recent work on high-dimensional bootstraps \citep[see][]{cherno2023hdbsreview}, we are not aware of existing bootstrap methods with inferential guarantees for our problem.

Empirically, in Figure \ref{fig::bsbad}(b), the bootstrap bias estimate ($\approx 0.06$) is three times smaller than the true bias ($\approx 0.17$). As a result, the bootstrap Z-statistics are highly inaccurate and have an average value of $\approx 25$ (while we would expect to see an average of $\approx 0$ if the test were performing well), leading to essentially a $100\%$ false positive rate.

\section{Methodology}\label{sec::methodology}

\subsection{Main idea}\label{subsec::mainidea}

As discussed in Section \ref{subsec::naive_perm_test}, the key challenge in developing a permutation test for $\mcH_0$ is that the OLS residual estimates $\hat\epsilon\ols$ do not satisfy the same independence properties as the true residuals $\epsilon$. Our solution is to introduce a new estimator $\hat\epsilon$ that exactly preserves some of the independence properties of $\epsilon$. To ease readability and build intuition, this subsection introduces the simplest possible variant of the mosaic permutation test. Section \ref{subsec::mpt} then introduces the mosaic permutation test in full generality. However, Section \ref{subsec::mpt} is self-contained, so readers may skip to Section \ref{subsec::mpt} if they wish.

For exposition, we make two simplifying assumptions for this subsection only: (1) the vectors of residuals $\epsilon_1, \dots, \epsilon_T \iid P_{\epsilon}$ are i.i.d. and (2) the exposures $L_t = L \in \R^{p \times k}$ do not change with time. Then, the main idea is to split the assets into two groups, $G_1, G_2 \subset [p]$, and estimate the residuals separately for each group. This ensures that under the null, the \textit{estimated} residuals for the assets in $G_1$ and $G_2$ are independent. Formally:
\begin{enumerate}[leftmargin=*, itemsep=0.75pt, topsep=0pt]
    \item Partition the assets into two disjoint groups $[p] = G_1 \cup G_2$. For now, $G_1 = \{1, \dots, \floor{p/2}\}$ and $G_2 = \{\floor{p/2} + 1, \dots, p\}$. (Section \ref{subsec::defaulttiling} discusses how to choose $G_1, G_2$.)
    \item For $i \in \{1,2\}$, let $\hat \epsilon_{t,G_i}$ denote the OLS estimate of $\epsilon_{t,G_i} \in \R^{|G_i|}$ based on $Y_{t,G_i} \in \R^{|G_i|}$, the returns of the assets in group $G_i$ at time $t$. Formally, let $H_i = (I_{|G_i|} - L_{G_i} (L_{G_i}^{\top} L_{G_i})^{-1} L_{G_i}^{\top})$ be the OLS projection matrix based on $L_{G_i}$, the exposures for the assets in $G_i$. Then we define $\hat \epsilon_{t,G_i} \defeq H_i Y_{t,G_i} = H_i \epsilon_{t,G_i}$.
    \begin{remark}     
    In each regression above, the parameters are the factor returns $X_t \in \R^k$ and the ``number of observations" is the number of stocks $|G_i|$ in group $i$. Splitting the assets into two groups halves the number of ``observations," leading to higher estimation error. Typically, the number of stocks is sufficiently large that we will still obtain reasonably good estimates $\hat \epsilon$. However, whenever $p < 2k$, this procedure will be powerless.
    \end{remark}
\end{enumerate}
Since the groups of stocks $G_1, G_2$ are disjoint, under $\mcH_0$, $\hat \epsilon_{t,G_1} = H_1 \epsilon_{t,G_1}$ and $\hat \epsilon_{t,G_2} = H_2 \epsilon_{t,G_2}$ are independent. Furthermore, $\{\epsilon_{t,G_i}\}_{t=1}^{T}$ are i.i.d. for each $i \in \{1,2\}$. Therefore, we can separately permute $\{\epsilon_{t,G_1}\}_{t=1}^{T}$ and $\{\epsilon_{t,G_2}\}_{t=1}^T$ without changing the law of $\hat \epsilon$:
\begin{equation}
\hat \epsilon \defeq 
    \begin{bmatrix} \hat \epsilon_{1,G_1} & \hat \epsilon_{1,G_2} \\ \hat \epsilon_{2,G_1} & \hat \epsilon_{2,G_2} \\ \vdots & \vdots \\ \hat \epsilon_{T,G_1} & \hat \epsilon_{T,G_2} 
     \end{bmatrix} 
\disteq 
    \begin{bmatrix} 
            \hat\epsilon_{\cblue{\pi_1}(1),G_1} & \hat \epsilon_{\cred{\pi_2}(1),G_2} \\ \hat\epsilon_{\cblue{\pi_1}(2),G_1} & \hat \epsilon_{\cred{\pi_2}(2),G_2} \\ 
            \vdots & \vdots \\  \hat\epsilon_{\cblue{\pi_1}(T),G_1} & \hat \epsilon_{\cred{\pi_2}(T),G_2} 
    \end{bmatrix} \defeq \hat \epsilon^{(\pi)}
\end{equation}
where $\pi_1, \pi_2 : [T] \to [T]$ are any permutations. The idea is best understood visually, and a color-assisted illustration of this equality is given below with $T=6$ observations. 
\begin{equation}\label{eq::pic1}
\hat\epsilon =  
\begin{tabular}{ V{4} C{0.15\textwidth} V{4} C{0.15\textwidth} V{4} } %{ |C{0.15\textwidth}|C{0.15\textwidth}| } 
 \multicolumn{2}{c}{$\overbrace{\hspace{0.275\textwidth}}^{\text{p assets split into two subsets}}$} \\ \hlineB{4}
 \colorone $\hat\epsilon_{1,G_1}$ & \coloronenew $\hat\epsilon_{1,G_2}$ \\ \hline  
 \colortwo $\hat\epsilon_{2,G_1}$ & \colortwonew $\hat\epsilon_{2,G_2}$ \\ \hline  
 \colorthree  $\hat\epsilon_{3,G_1}$ & \colorthreenew $\hat\epsilon_{3,G_2}$ \\ \hline  
 \colorfour $\hat\epsilon_{4,G_1}$ & \colorfournew $\hat\epsilon_{4,G_2}$ \\ \hline  
 \colorfive $\hat\epsilon_{5,G_1}$ & \colorfivenew $\hat\epsilon_{5,G_2}$ \\ \hline  
  % $\vdots$ & $\vdots$ \\ \hline 
 \colorsix $\hat\epsilon_{6,G_1}$ & \colorsixnew $\hat\epsilon_{6,G_2}$ \\ \hlineB{4}  
 \multicolumn{2}{c}{} \\
 %\multicolumn{2}{c}{} 
\end{tabular}
\,\,\,\disteq\,\,\, 
\begin{tabular}{ V{4} C{0.15\textwidth} V{4} C{0.15\textwidth} V{4} }
 % \multicolumn{2}{c}{$\overbrace{\tenspace\tenspace\tenspace}^{\text{$\hat \epsilon_{\cdot,G_1}$ is unchanged}}\tenspace\overbrace{\tenspace\tenspace\tenspace}^{\text{$\hat \epsilon_{\cdot,G_2}$ is permuted}}$} \\ \hline  
  \multicolumn{2}{c}{$\overbrace{\hspace{0.275\textwidth}}^{\text{ $\hat\epsilon_{\cdot,G_1}, \hat\epsilon_{\cdot,G_2}$ are separately permuted}}$} \\ \hlineB{4}
 \colorfive $\hat\epsilon_{5,G_1}$ & \colorthreenew $\hat\epsilon_{3,G_2}$ \\ \hline  
 \colorone $\hat\epsilon_{1,G_1}$ & \colorsixnew $\hat\epsilon_{6,G_2}$ \\ \hline  
 \colorfour  $\hat\epsilon_{4,G_1}$ & \coloronenew $\hat\epsilon_{1,G_2}$ \\ \hline  
 \colorthree $\hat\epsilon_{3,G_1}$ & \colortwonew $\hat\epsilon_{2,G_2}$ \\ \hline  
 \colorsix $\hat\epsilon_{6,G_1}$ & \colorfivenew $\hat\epsilon_{5,G_2}$ \\ \hline  
 % $\vdots$ & $\vdots$ \\ \hline
 \colortwo $\hat\epsilon_{2,G_1}$ & \colorfournew $\hat\epsilon_{4,G_2}$ \\ \hlineB{4}  
 \multicolumn{2}{c}{} \\
\end{tabular}
\end{equation}
Above, the two groups are shown in different colors, and the shading of each cell denotes its original position in time---for this reason, the shadings in the right panel indicate that $\hat \epsilon_{\cdot,G_1}$ and $\hat \epsilon_{\cdot,G_2}$ have been separately permuted. 

After sampling permutations uniformly at random to create new estimated residual matrices $\tilde \epsilon^{(1)}, \dots, \tilde \epsilon^{(R)}$, we can compute a valid p-value using any test statistic $S : \R^{T \times p} \to \R$:
\begin{equation}\label{eq::hatepspval}
    \pval \defeq \frac{1 + \sum_{r=1}^R \I(S(\hat \epsilon) \le S(\tilde \epsilon^{(r)}))}{R+1}.
\end{equation}
For example, $S(\hat \epsilon)$ could measure the maximum absolute correlation of a (regularized) estimate of the covariance matrix of $\epsilon_t$. We discuss the choice of test statistic in more detail in Sections \ref{sec::realdata} and \ref{sec::extensions}. The key result is that this p-value is valid for any test statistic.

\begin{theorem}\label{thm::validitybasic} $\pval$ in Eq. \ref{eq::hatepspval} is an exact p-value testing $\mcH_0$ assuming $\epsilon_1, \dots, \epsilon_T \iid P_{\epsilon}$ are i.i.d.
\end{theorem}

This simple procedure already has many desirable properties, but it lacks stability due to the choice of $G_1$ and $G_2$. For example, if the residuals of stocks $1$ and $2$ are highly correlated, we will have no power to detect this if both stocks are placed in the same group, since their estimated idiosyncratic returns $\hat \epsilon_{\cdot,1}, \hat \epsilon_{\cdot,2}$ will never be ``separated" by different permutations. We address this problem in the next section.

\subsection{The mosaic permutation test}\label{subsec::mpt}

We now introduce the general mosaic permutation test, which is more powerful and stable than the simple method in Section \ref{subsec::mainidea}. As an added benefit, we will also make the test more robust to autocorrelation and nonstationarity among the residuals $\epsilon_1, \dots, \epsilon_T$.

In Section \ref{subsec::mainidea}, we separated the asset returns $\bY \in \R^{T \times p}$ into two disjoint groups, computed residual estimates $\hat \epsilon$ separately for each group, and then permuted within each group. Now, we partition the data $\bY$ into an arbitrary number $M$ of rectangles along both axes. Formally, for $m=1,\dots,M$, let $B_m \subset [T]$ denote a subset or ``batch" of observations and $G_m \subset [p]$ denote a subset or ``group" of assets. We say $\{(B_m, G_m)\}_{m=1}^M$ is a \textit{tiling} if for every timepoint $t$ and asset $j$, there is exactly one pair $(B_m, G_m)$ such that $t \in B_m$ and $j \in G_m$. See Figure \ref{fig::tilingex} for an illustration of this definition.

\begin{figure}[!h]
    \centering
    \includegraphics[width=\linewidth]{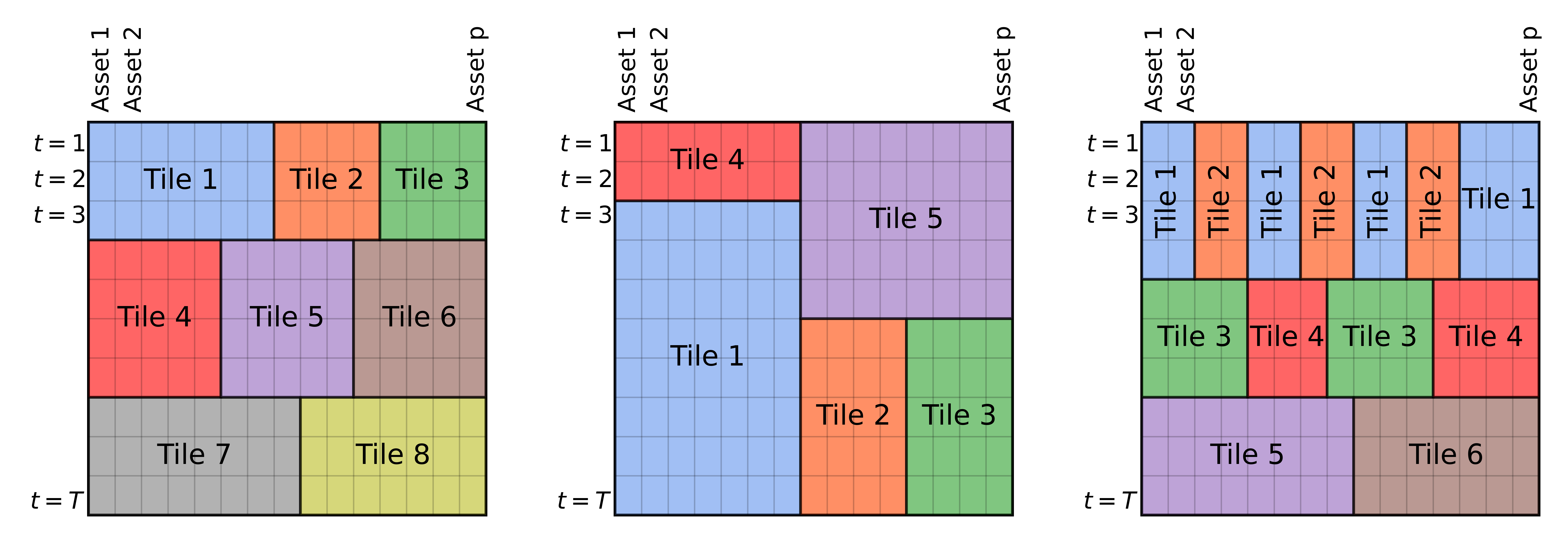}
    \caption{This figure shows three examples of tilings of the data matrix $\bY$. The right-most example emphasizes that each tile need not be contiguous in the initial ordering of the features (although in illustrations, we often display examples of contiguous tiles since they are easier to plot). Note that in our default choice of tiling (Section \ref{subsec::defaulttiling}), the initial ordering of the features does not affect the results of the test.}
    \label{fig::tilingex}
\end{figure}

We will soon discuss how to choose the tiling (in Section \ref{subsec::defaulttiling}). For now, given an arbitrary tiling, we refer to the submatrices $\bY_{(m)} = \bY_{B_m, G_m}$ and $\epsilon_{(m)} \defeq \epsilon_{B_m,G_m}$ as the $m$th \emph{tiles} of the matrices $\bY$ and $\epsilon$, respectively. Before presenting the mosaic permutation test, we make two assumptions, generalizing the assumptions from Section \ref{subsec::mainidea}.

First, in Section \ref{subsec::mainidea}, we assumed that the exposures $L_t$ did not change with time. Now, we ask that the analyst chooses the tiles such that the exposure matrices $\{L_t\}_{t=1}^T$ may change across tiles, but not within tiles.
For example, in our application to the BFRE model, the exposures change every week. As a result, we choose the tiles such that each tile contains data from only one week. If $L_t$ takes unique values at every observation, then testing $\mcH_0$ is possible but requires an extension introduced in Section \ref{subsec::nonconst_lt}.

\begin{assumption}\label{assump::batched_exposure} The exposures $\{L_t\}_{t \in B_m}$ in each tile are all equal.
\end{assumption}

Second, in Section \ref{subsec::mainidea}, we assumed that the idiosyncratic returns $\epsilon_1, \dots, \epsilon_T$ were i.i.d. We now relax this assumption to allow for a large degree of nonstationarity among the residuals. In particular, we assume that each asset's returns are exchangeable \textit{within} tiles, but not necessarily between tiles.

\begin{assumption}[Local exchangeability]\label{assump::localexch} For every asset $j \in [p]$, we assume the following. Let $\pi : [T] \to [T]$ be any permutation such that $\bY_{t,j}$ and $\bY_{\pi(t), j}$ are always in the same tile. Then
\begin{equation}
    (\epsilon_{1,j}, \dots, \epsilon_{T,j}) \disteq
    (\epsilon_{\pi(1),j}, \dots, \epsilon_{\pi(T),j}).
\end{equation}
\end{assumption}

Assumption \ref{assump::localexch} allows the distribution of the residuals to drift between tiles, making this assumption much weaker than the i.i.d. assumption in Theorem \ref{thm::validitybasic}. Indeed, this assumption is related to the motivation for many classical procedures for time series data, such as the block bootstrap \citep{kunsch1989block_bootstrap}. Armed with these assumptions, Algorithm \mosaicpermalgnum\, defines the mosaic permutation test.

\hrulefill 
\addtocounter{algorithm}{1}

\vspace{-0.3cm}
\textbf{Algorithm 1}: The mosaic permutation test. 

\vspace{-0.425cm}
\hrulefill

\textbf{Inputs}: Asset returns $\bY \in \R^{T \times p}$, exposures $L_t \in \R^{p \times k}$ for $t \in [T]$, tiles $\{(B_m, G_m)\}_{m=1}^M$ and a test statistic $S : \R^{T \times p} \to \R$. 

\begin{steps}[leftmargin=*, topsep=0.5pt, itemsep=0.5pt]
\item For each tile $m=1,\dots,M$, we let $\hat \epsilon_{(m)}$ denote the OLS estimate of $\epsilon_{(m)}$ using only the data in $\bY_{(m)}$. Precisely, let $L_{(m)} \in \R^{|G_m| \times k}$ denote the exposures for the assets in the $m$th tile (note by Assumption \ref{assump::batched_exposure} that the exposures are constant over time within the tile). Let $H_m \defeq (I_{|G_m|} - L_{(m)} (L_{(m)}\trans L_{(m)})^{-1} L_{(m)}\trans)$ be the OLS projection matrix based on $L_{(m)}$. Then we define
    \begin{equation}\label{eq::hateps_tile_def}
        \hat\epsilon_{(m)} \defeq \bY_{(m)} H_m = \epsilon_{(m)} H_m
    \end{equation}
and $\hat \epsilon \in \R^{T \times p}$ is defined such that $\hat\epsilon_{B_m, G_m} \defeq \hat\epsilon_{(m)}$ for $m \in [M]$.
    
\item For $r=1, \dots, R$, randomly reorder the rows within each tile and let $\tilde{\epsilon}^{(r)} \in \R^{T \times p}$ denote the resulting matrix. Below, this is illustrated with $T=7$ observations and $M=4$ tiles:
\begin{equation}\label{eq::full_picture}
\hat\epsilon = 
\begin{tabular}{ V{4}C{0.1\textwidth}V{4}C{0.1\textwidth}V{4}C{0.1\textwidth}V{4} } 
 \multicolumn{3}{c}{$\overbrace{\hspace{0.275\textwidth}}^{\text{partitioned into $M=4$ tiles}}$} \\ 
 \hlineB{4}
\multicolumn{2}{V{4} C{0.2\textwidth} V{4}}{\oone $\hat\epsilon_{1,G_1}$} & \tone $\hat\epsilon_{1,G_2}$ \\ \hline
\multicolumn{2}{V{4} C{0.2\textwidth} V{4}}{\otwo $\hat\epsilon_{2,G_1}$} & \ttwo $\hat\epsilon_{2,G_2}$ \\ \hline 
\multicolumn{2}{V{4} C{0.2\textwidth} V{4}}{\othree $\hat\epsilon_{3,G_1}$} & \tthree $\hat\epsilon_{3,G_2}$ \\ \hline 
\multicolumn{2}{V{4} C{0.2\textwidth} V{4}}{\ofour $\hat\epsilon_{4,G_1}$} & \tfour $\hat\epsilon_{4,G_2}$ \\ \hlineB{4} 
\thone $\hat\epsilon_{5,G_3}$ & \multicolumn{2}{C{0.2\textwidth} V{4}}{\fone $\hat\epsilon_{5,G_4}$} \\ \hline  \thtwo $\hat\epsilon_{6,G_3}$ & \multicolumn{2}{C{0.2\textwidth} V{4}}{\ftwo $\hat\epsilon_{6,G_4}$} \\ \hline  
\ththree $\hat\epsilon_{7,G_3}$ & \multicolumn{2}{C{0.2\textwidth} V{4}}{\fthree $\hat\epsilon_{7,G_4}$} \\ \hlineB{4} 
\multicolumn{3}{c}{} \\ 
\end{tabular}
\,\,\,\,\disteq\,\,\,\, 
\begin{tabular}{ V{4}C{0.1\textwidth}V{4}C{0.1\textwidth}V{4}C{0.1\textwidth}V{4} }
 \multicolumn{3}{c}{$\overbrace{\hspace{0.275\textwidth}}^{\text{separately permute each tile}}$} \\ 
 \hlineB{4}
\multicolumn{2}{V{4} C{0.2\textwidth} V{4}}{\ofour $\hat\epsilon_{4,G_1}$} & \ttwo $\hat\epsilon_{2,G_2}$ \\ \hline
\multicolumn{2}{V{4} C{0.2\textwidth} V{4}}{\othree $\hat\epsilon_{3,G_1}$} & \tone $\hat\epsilon_{1,G_2}$ \\ \hline 
\multicolumn{2}{V{4} C{0.2\textwidth} V{4}}{\oone $\hat\epsilon_{1,G_1}$} & \tfour $\hat\epsilon_{4,G_2}$ \\ \hline 
\multicolumn{2}{V{4} C{0.2\textwidth} V{4}}{\otwo $\hat\epsilon_{2,G_1}$} & \ttwo $\hat\epsilon_{3,G_2}$ \\ \hlineB{4} 
\ththree $\hat\epsilon_{7,G_3}$ & \multicolumn{2}{C{0.2\textwidth} V{4}}{\ftwo $\hat\epsilon_{6,G_4}$} \\ \hline  
\thone $\hat\epsilon_{5,G_3}$ & \multicolumn{2}{C{0.2\textwidth} V{4}}{\fthree $\hat\epsilon_{7,G_4}$} \\ \hline  
\thtwo $\hat\epsilon_{6,G_3}$ & \multicolumn{2}{C{0.2\textwidth} V{4}}{\fone $\hat\epsilon_{5,G_4}$} \\ \hlineB{4}  
\multicolumn{3}{c}{} \\ 
\end{tabular}
\end{equation}

Mathematically, let $P_m^{(r)} \in \R^{|B_m| \times |B_m|}$ denote a uniformly random permutation matrix for $m \in [M], r \in [R]$. Then $\tilde\epsilon^{(r)}$ is defined such that $\tilde{\epsilon}^{(r)}_{B_m, G_m} \defeq P_m^{(r)} \hat \epsilon_{(m)}$ for $m \in [M]$.
    
\item For any test statistic $S : \R^{T \times p} \to \R$, compute the p-value 
\begin{equation}\label{eq::hatepspval2}
    \pval \defeq \frac{1 + \sum_{r=1}^R \I(S(\hat \epsilon) \le S(\tilde \epsilon^{(r)}))}{R+1}.
\end{equation}
\end{steps}
\hrulefill

\begin{remark}\label{rem::nobs2} If the $m$th tile contains data from $t_m = |B_m|$ timepoints, then Eq. \ref{eq::hateps_tile_def} is equivalent to running $t_m$ separate cross-sectional regressions to compute each row of $\hat\epsilon_{(m)}$. I.e., $\hat \epsilon_{t,G_m} \defeq H_m Y_{t,G_m}$ for all $t \in B_m$. Thus, the ``number of observations" is the number of stocks in $G_m$, and the parameters are the factor returns $X_t \in \R^k$.
\end{remark}

Figure \ref{fig::intro_figure} gives a simple illustration of the mosaic permutation test with $M=10$ tiles; note that the method in Section \ref{subsec::mainidea} is an example of this procedure with $M=2$ tiles. Theorem \ref{thm::main} states that $\pval$ is a valid p-value assuming only Assumptions \ref{assump::batched_exposure}-\ref{assump::localexch}. We emphasize that these assumptions allow for (i) the residuals $\epsilon$ and factors $\{X_t\}_{t=1}^T$ to be arbitrarily heavy-tailed, (ii) the factors $\{X_t\}_{t=1}^T$ to be arbitrarily nonstationary and autocorrelated, and (iii) the residuals $\{\epsilon_t\}_{t=1}^T$ to be nonstationary between batches, and (iv) the use of any test statistic $S(\cdot)$.

\begin{theorem}\label{thm::main} Suppose Assumptions \ref{assump::batched_exposure}-\ref{assump::localexch} hold. Then under $\mcH_0$, Eq. \ref{eq::hatepspval2} defines a valid p-value satisfying $\P(\pval \le \alpha) \le \alpha$ for any $\alpha \in (0,1)$. 
\end{theorem}

Theorem \ref{thm::main} is proved in Appendix \ref{appendix::proofs}. However, the main idea is that the residuals for the tiles $\{\hat\epsilon_{(m)}\}_{m=1}^M$ are estimated using separate data. Thus, separately reordering the rows within each tile does not change the joint law of the full estimated residual matrix $\hat\epsilon$. Formally, if $P_1 \in \{0,1\}^{|B_1| \times |B_1|}, \dots, P_M \in \{0,1\}^{|B_M| \times |B_M|} $ are permutation matrices, 
\begin{equation}\label{eq::full_disteq}
    (\hat\epsilon_{(1)}, \dots, \hat \epsilon_{(M)}) \disteq (P_1 \hat\epsilon_{(1)}, \dots, P_M \hat \epsilon_{(M)}).
\end{equation}

\begin{remark}[Regularization]\label{rem::regularization} Theorem \ref{thm::main} allows the use of arbitrary test statistics $S(\hat\epsilon)$, including regularized estimates of $\cov(\epsilon)$. However, we require the use of unregularized OLS regressions in each tile to estimate $\hat\epsilon$. The reason is that regularized (e.g.) ridge estimates of the residuals $\epsilon$ will not project out the influence of the factor returns $\{X_t\}_{t \in [T]}$, causing violations of Eq. \ref{eq::full_disteq} since all tiles may depend somewhat on the factor returns $\{X_t\}_{t \in [T]}$. 
That said, incorporating regularized estimates of the residuals is a promising direction for future work (see Section \ref{sec::discussion}). 
\end{remark}

\subsection{A default choice of tiling}\label{subsec::defaulttiling}

We recommend choosing the tiling by doing the following:
\begin{enumerate}[topsep=1pt, itemsep=0.5pt, leftmargin=*]
    \item First, partition the time points into $[T] = B_1 \cup \dots \cup B_{I}$ into $I \in \N$ batches. By default, we take each batch to be $10$ consecutive time points, so $B_1 = \{1, \dots, 10\}, B_2 = \{11, \dots, 20\}$, etc. That said, if necessary, one may make the batches smaller to guarantee that the exposures are constant within each batch (satisfying Assumption \ref{assump::batched_exposure}).
    \item For each batch $i=1,\dots,I$, randomly partition the assets into $D$ groups $[p] = G_{i,1} \cup \dots \cup G_{i,D}$ of (roughly) equal size. We recommend setting $D = \max\left(2, \ceil{\frac{p}{5k}}\right)$.

    \item We let $\{(B_i, G_{i,k}) : d \in [D], i \in [I]\}$ be the final set of $M = I \cdot D$ tiles. For example, in Figure \ref{fig::tilingex}, the tilings in the left and right panels are of this form (while the one in the middle panel is not).
\end{enumerate}

Above, we use small batch sizes of $\approx 10$ observations because smaller contiguous batches are more robust to nonstationarity and autocorrelation among the residuals. Furthermore, using more batches increases the stability of the test and decreases the likelihood that any one random partition of the assets dramatically affects the results. 

Furthermore, the choice of $D$ balances the following trade-off. On the one hand, using larger $D$ increases the probability that any two assets are in different tiles and thus that their estimated idiosyncratic returns can be ``separated" by permutations, since there is a $\approx \frac{1}{D}$  chance that any two assets lie in the same group. On the other hand, as per Remark \ref{rem::nobs2}, we must separately estimate the value of the factors $X_t \in \R^k$ within each tile using a linear regression with $\frac{p}{D}$ observations. Increasing $D$ will reduce the number of observations per regression and increase the estimation error of $\hat \epsilon$. Thus, choosing $D = \max(2, \floor{\frac{p}{5k}})$ maximizes the number of tiles subject to the constraint that there are $5$ times as many observations as there are parameters in each regression used to estimate $\hat \epsilon$.

\section{Application to the BFRE model}\label{sec::realdata}

We now apply our method to the Blackrock Fundamental Equity Risk (BFRE) model introduced in Section \ref{subsec::bfre_motivation}. We ask whether the BFRE model accurately describes correlations among asset returns in three sectors of the US economy: energy, financials, and healthcare. Appendix \ref{appendix::othersectors} also analyzes six additional sectors; however, the results below are representative of our findings in the appendix.

Overall, the BFRE model appears to explain the majority of correlations among assets. However, we find evidence of a missing persistent factor among real estate stocks and a missing transient factor in healthcare post-COVID. We do not find consistent evidence against the null in the energy sector. These findings are supported by four analyses:
\begin{itemize}[noitemsep, topsep=0pt, leftmargin=*]
    \item Our main analysis straightforwardly applies the mosaic permutation test to each sector.
    \item Our second analysis finds in an ablation test that removing a fraction of the existing factors from the BFRE model leads to strong evidence against the null. 
    \item Our third analysis analyzes the degree to which we can improve the BFRE model.
    \item Our final analysis delves into our results in the financial sector, where we find evidence of correlations among certain real estate asset residuals. 
\end{itemize}

\subsection{Main analysis}\label{subsec::main_bfre}

We now present our main analysis of the BFRE model. Methodologically, we use the default choice of tiling discussed in Section \ref{subsec::defaulttiling}, and we use the mean maximum absolute correlation (MMC) test statistic from Section \ref{subsec::motivation}. We use this test statistic because it is simple and interpretable, although we shall soon discuss other choices. 

Figure \ref{fig::motivation_soln} plots the test statistic and significance threshold for the energy, financial, and healthcare sectors based on sliding windows of $350$ observations. To assess significance, Figure \ref{fig::mmc_pval} plots two types of Z-statistics. First, it shows the exact mosaic Z-statistic $Z_{\mathrm{exact}} = \Phi^{-1}(1-\pval)_+$, where $\Phi^{-1}$ denotes the inverse CDF of a standard normal distribution. $Z_{\mathrm{exact}}$ is stochastically dominated by the positive part of a standard normal under the null, so large values of $Z$ are evidence against the null. However, since $p_{\val}$ is discrete, $Z_{\mathrm{exact}}$ is bounded below a maximum value which increases slowly with the number of randomizations (in our analysis, $Z_{\mathrm{exact}} \le 3.24$). Thus, Figure \ref{fig::mmc_pval} also plots the approximate Z-statistic:
\begin{equation}\label{eq::apprx_Z}
    Z\apprx = \frac{S(\hat\epsilon) - \bar S}{\sqrt{\frac{1}{R+1} \sum_{r=0}^R (S(\tilde\epsilon^{(r)}) - \bar S)^2}} \text{ for } \bar S = \sum_{r=0}^{R+1} S(\tilde\epsilon^{(r)}),
\end{equation}
where above, $\hat\epsilon$ are mosaic residuals, $\tilde\epsilon^{(1)}, \dots, \tilde\epsilon^{(R)}$ are permuted variants of the mosaic residuals, and we use the convention $\tilde\epsilon^{(0)} \defeq \hat\epsilon$. $Z\apprx$ may not be exactly Gaussian, but Lemma \ref{lem::apprx_Z} in Appendix \ref{appendix::apprx_Z} proves that it has zero mean and unit variance under the null. Furthermore,  $Z\apprx$ can take large values without excessive computation.

\begin{figure}[!h]
    \centering
    \includegraphics[width=\linewidth]{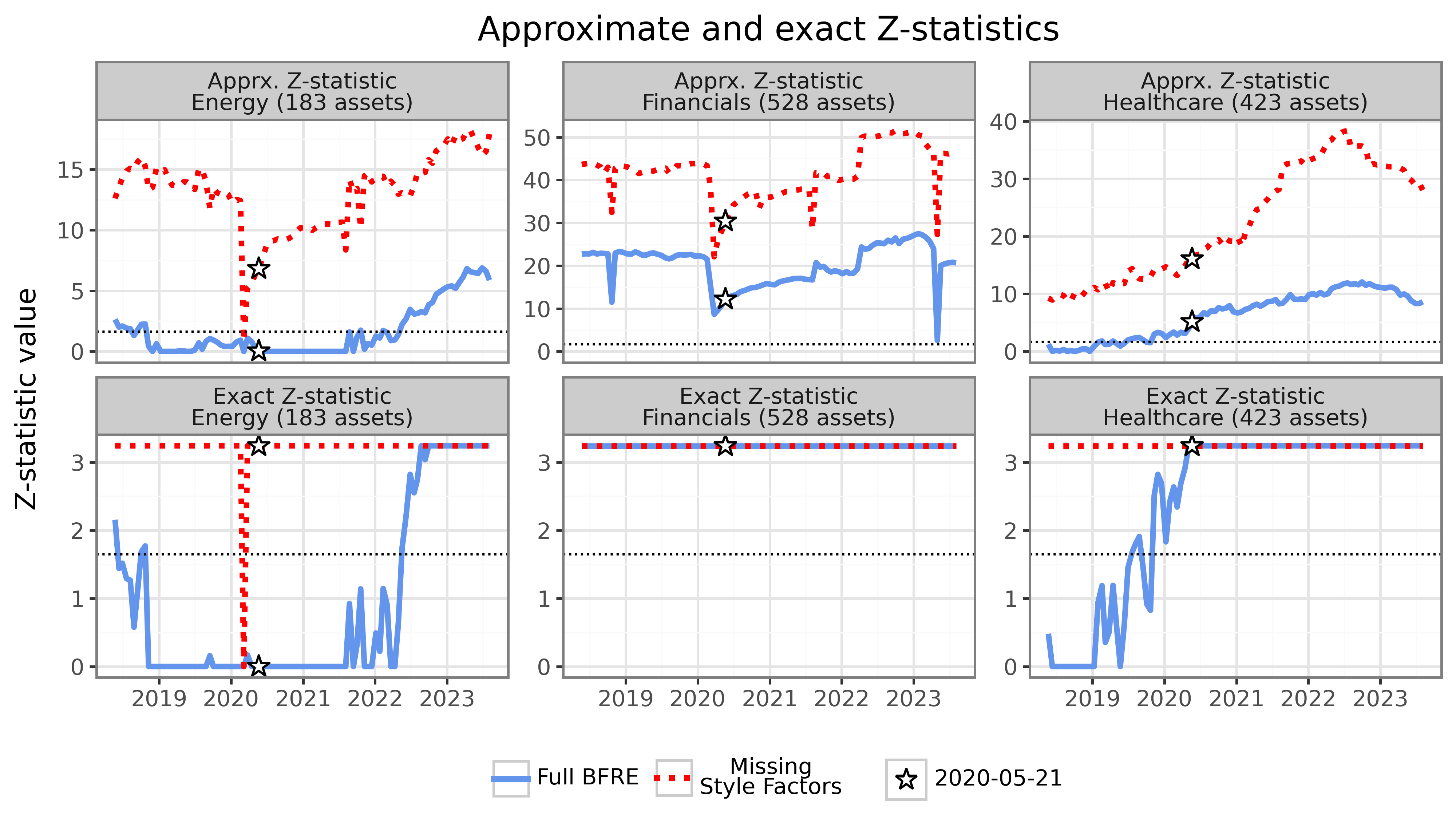}
    \caption{This figure plots exact and approximate mosaic Z-statistics computed for the analysis in Figure \ref{fig::motivation_soln} and the ablation analysis in Figure \ref{fig::ablation}. The exact Z-statistics $\Phi^{-1}(1-\pval)_+$ are stochastically dominated by the positive part of a standard normal under the null, but their maximum value is $\approx 3.24$, corresponding to the minimum adjusted p-value of $\frac{3}{5001}$ when using $R=5000$ randomizations. In contrast, the approximate Z-statistics from Eq. \ref{eq::apprx_Z} may not be Gaussian; however, they are provably mean zero and have unit variance under the null (see Appendix \ref{appendix::apprx_Z}). For the exact Z-statistics, we apply a Bonferroni correction for the $m=3$ sectors, but we do not apply a multiplicity correction across all $100$ time points. The dotted black line denotes the significance threshold $(1.64)$ for $\alpha=0.05$. There are $k=20, k=27$, and $k=27$ factors in energy, financials, and healthcare, respectively, although a few factors are not used at all time points.}\label{fig::mmc_pval}
\end{figure}

The findings are qualitatively different for different sectors. In the energy sector, we generally do not find statistically significant evidence against the null before mid-2022. In contrast, in the healthcare sector, we primarily find evidence against the null after the COVID-19 pandemic began, suggesting that we may be detecting a transient factor. In the financial sector, we find highly statistically significant evidence of violations of the null $\mcH_0$ at all time points, and the resulting p-values are highly significant $(p <0.001)$. That said, it is not clear if this represents a large effect size. Heuristically, the test statistic value is often close to the significance threshold (in healthcare and financials, respectively, the difference is never larger than $3\%$ and $8\%$). To investigate this further, we perform three additional analyses in the next three sections.

\begin{remark} We stress that failing to reject the null $\mcH_0$ does not confirm that the null holds. For example, a different test statistic might yield stronger evidence against the null in the healthcare sector pre-COVID (although we are unaware of such a statistic). However, our analysis illustrates that the mosaic method can assess whether any given test statistic provides evidence against the null. 
\end{remark}

\begin{remark}
The energy sector contains fewer assets than the other sectors, and we should expect lower power in smaller sectors (since having fewer assets can lead to lower-quality residual estimates). However, Appendix \ref{appendix::othersectors} analyzes two other sectors of comparable size, where we find highly significant evidence against the null. Thus, our results in the energy sector are not explained by its size.
\end{remark}

\begin{remark} Our dataset has essentially no missing data, with the exception that some assets did not exist for parts of our analysis. E.g., ``Oak Street Health" was not publicly traded until 2020.
Despite this, Appendix \ref{appendix::proofs} details a simple strategy which yields an exact p-value as long as each asset's residuals are locally exchangeable (Assumption \ref{assump::localexch}) conditional on the data availability pattern. Essentially, it suffices to choose tiles such that no tile contains any missing data. 
Please see Appendix \ref{subsec::addnmethoddetails} for further methodological details.
\end{remark}

\subsection{Ablation test}

We now perform an ablation test that removes a subset of factors from the BFRE model. Indeed, the BFRE model contains two types of factors: (i) binary \textit{industry} factors, where asset $j$'s exposure to (e.g.) the banking industry factor is the binary indicator of whether asset $j$ is a bank, and (ii) twelve continuous \textit{style} factors, such as ``size," where a corporation's exposure to the ``size" factor is a measure of the size of the corporation.

We now repeat our analysis of the BFRE model after removing the style factors. Figure \ref{fig::mmc_pval} shows that the mosaic Z-statistics become more significant---in fact, the discrete mosaic p-values are nearly always equal to their minimum possible value ($< 0.001$) in this analysis. Furthermore, Figure \ref{fig::ablation} shows that the differences between the mosaic test statistics and their significance thresholds double or more compared to the analysis of the full BFRE model. Overall, this gives more (heuristic) evidence that the BFRE model accounts for the most significant correlations among assets, even though it does not fit perfectly. Soon, we will see more evidence for this claim.

\begin{figure}
    \centering
    \includegraphics[width=\linewidth]{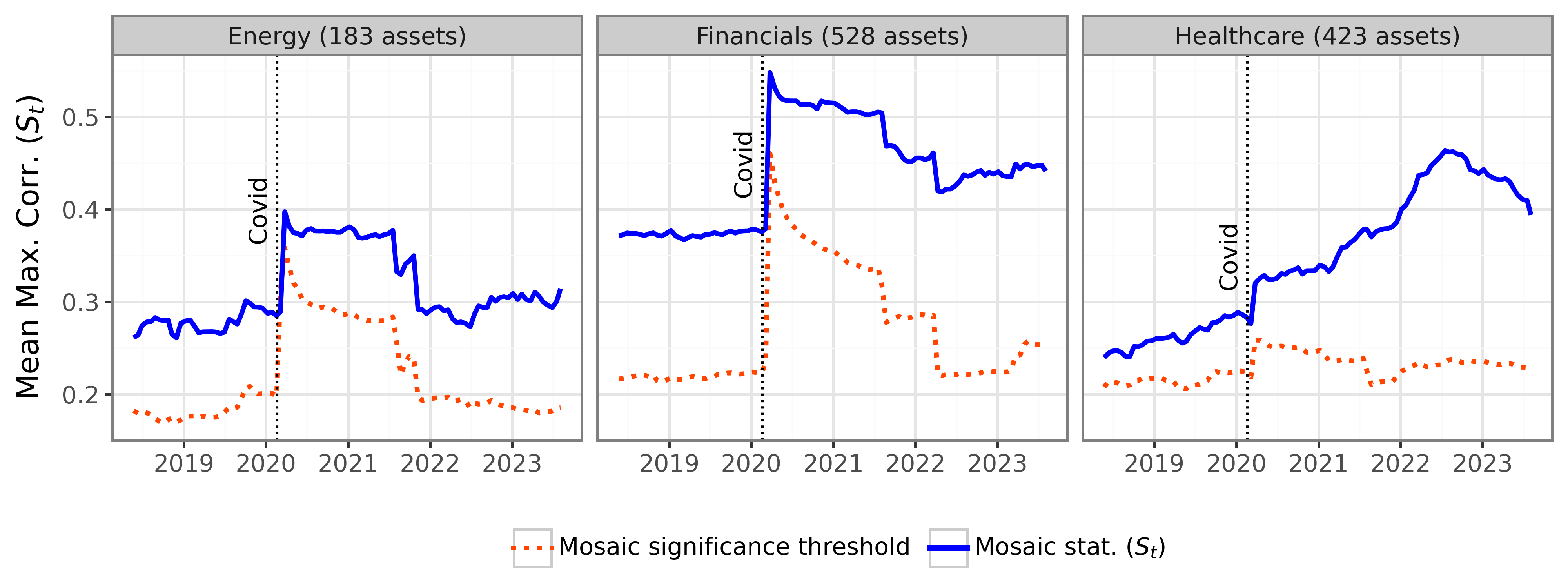}
    \caption{This figure repeats the analysis in Figure \ref{fig::motivation}, except now we perform an ablation test where we remove the twelve style factors from the BFRE model.
    }
    \label{fig::ablation}
\end{figure}

\subsection{Improving the model}\label{subsec::improvement}

The main contribution of our paper is a method for testing the goodness-of-fit of a factor model. That said, another important question is whether one can \textit{improve} a model. We now illustrate how the mosaic permutation test can help answer this question. Consider a simple alternative where the BFRE model is missing at least one component:
\begin{equation}
    Y_t = L_t X_t + v Z_t + \epsilon_t,
\end{equation}
where $L_t \in \R^{p \times k}$ are the pre-existing BFRE factor exposures, $X_t$ are the BFRE factors, $Z_t \in \R$ is an additional factor, and $v \in \R^p$ is an unknown (missing) factor exposure. We ask: can we estimate a new exposure $v$ so that the model fit is improved? 

To measure model performance, we split the data into two folds. On the first fold, we estimate new exposures $\hat v$. On the second fold, we check whether $\hat v$ allows us to better predict each residual $\hat \epsilon_{t,j}$ from the residuals in the other tiles. (Under the null, $\hat\epsilon_{t,j}$ is independent of the residuals of all assets in different tile from $j$ at time $t$, so this prediction task is hopeless unless the null is violated.) Finally, we check statistical significance using the mosaic permutation test. The details of this analysis are described below.

\underline{Step 1: estimating $v$}. For each sector, we construct several estimators $\hat v_0, \dots, \hat v_I$ of $v$ using $T_0 \approx 1250$ observations from January 2016 until January 2021. If $\hat C \in \R^{p \times p}$ denotes the empirical correlation matrix of the mosaic residual estimates $\{\hat\epsilon_t\}_{t=1}^{T_0}$, our first estimate $\hat v_0$ is the top eigenvector of $\hat C$. However, sparse estimates of $v$ could yield higher power. Thus, we also approximately solve a sparse PCA objective for various sparsity levels $\ell$:
\begin{equation}\label{eq::sparse_pca}
    \hat v \approx \argmax_{\|v\|_2 = 1} v^T \hat C v \text{ s.t. } \|v\|_0 \le \ell,
\end{equation}
where $\|v\|_0 = |\{i \in [p] : v_i \ne 0\}|$ denotes the number of nonzero entries of $v$. Exactly solving this optimization problem is computationally intractable, so we use a simple greedy approximation (Appendix \ref{appendix::greedy_sparse_pca}). We compute estimators $\hat v_1, \dots, \hat v_I$ for $I=10$ values of $\ell$, evenly spaced between $20$ and $p$.

\underline{Step 2: measuring out-of-sample performance.} For each $\hat v_i$, we now estimate the out-of-sample error of the new model. Formally, for any time point $t > T_0$ and asset $j$ in the $m$th tile, we check whether $\hat v_i$ allows us to predict the mosaic residual $\hat\epsilon_{t,j}$ from $\hat\epsilon_{t,-G_m}$, the set of residuals at time $t$ which are not in the same tile as $\hat\epsilon_{t,j}$. To do this, we first estimate $Z_t$, the missing factor return, using $\hat\epsilon_{t,-G_m}$:
$$\hat Z_t^{(i,j)} = \frac{\hat\epsilon_{t,-G_m}^T \hat v_{i,-G_m}}{\|\hat v_{i,-G_m}\|_2^2} \in \R.$$
Then, we set $\hat \gamma_{t,j}^{(i)} \defeq \hat v_{i,j} \cdot \hat Z_t^{(i,j)}$ to be the out-of-sample OLS estimate of $\hat\epsilon_{t,j}$ based on $\hat Z_t^{(i,j)}$. Note that $\hat \gamma_{t,j}^{(i)}$ depends only on $\hat \epsilon_{t,-G_m}$ and the first fold of the dataset, so if the two folds of the dataset are independent and $\mcH_0$ holds, then $\hat\gamma_{t,j}^{(i)}$ is independent of $\hat\epsilon_{t,j}$. In contrast, if $\mcH_0$ is violated and $\hat v_i$ explains additional correlations among the assets, then $\hat\gamma_{t,j}^{(i)}$ should predict $\hat\epsilon_{t,j}$. As an aggregate measure of model performance, we compute the out-of-sample $R^2$ of these new predictions:
\begin{equation*}
    \hat r_i^2 = 1 - \frac{\sum_{t=T_0+1}^T \sum_{j=1}^p \left(\hat\gamma_{t,j}^{(i)} - \hat \epsilon_{t,j} \right)^2}{\sum_{t=T_0+1}^T \sum_{j=1}^p \hat\epsilon_{t,j}^2}.
\end{equation*}

Inspired by \cite{owen2016bicross, owen2009bicross}, we refer to this as a mosaic \textit{bi-cross validation} (BCV) $R^2$. Our final test statistic is the maximum $R^2$ over all sparse-PCA estimates $\hat v_0, \dots, \hat v_I$:
\begin{equation}
    S = \max(\hat r_0^2, \dots, \hat r_I^2).
\end{equation}
By taking the maximum out-of-sample $R^2$ values, we hope to gain power no matter the underlying sparsity level of $v$. We compute this statistic in sliding windows of up to $350$ observations on the second fold of data, and we check for statistical significance using the mosaic permutation test.

Figure \ref{fig::r2_plot} shows that the maximum BCV $R^2$ value is negative and statistically insignificant for the energy sector, and it is positive (usually $\approx 1\%$) and highly significant for financials. In line with our prior analyses, the test statistic value is small, suggesting that adding an estimated factor only slightly improves the model. Interestingly, the maximum BCV $R^2$ statistic is \textit{negative} but occasionally statistically significant for the healthcare sector. In other words, adding an extra estimated factor can \textit{degrade} model performance, but under the null, we would expect to see an even larger degradation. This result could be a false positive due to multiplicity; however, it also aligns with prior theoretical analysis of PCA and factor models. For example, \cite{paul2007} showed that when testing whether a covariance matrix equals the identity, if the population maximum eigenvalue is slightly larger than $1$, the empirical maximum eigenvalue can often be statistically significant even while the empirical maximum eigenvector is approximately orthogonal to the population maximum eigenvector. Similarly, \cite{owen2016bicross} emphasized that incorporating a missing factor exposure with a small effect size can degrade out-of-sample predictions by increasing their variance. Overall, this analysis suggests that the unexplained correlations among healthcare assets are transient, reasonably small, or both. 

We also continue our ablation test and repeat this analysis after removing the style factors from the BFRE model. As shown by Figure \ref{fig::r2_plot}, the maximum BCV $R^2$ values become larger and are uniformly highly significant. Interestingly, for the energy sector, the maximum BCV $R^2$ values are still small ($\le 1\%$); we leave the question of why to future work. Nonetheless, this analysis gives more evidence that the BFRE model explains the most significant correlations among assets.

\begin{figure}
    \centering
    \includegraphics[width=\linewidth]{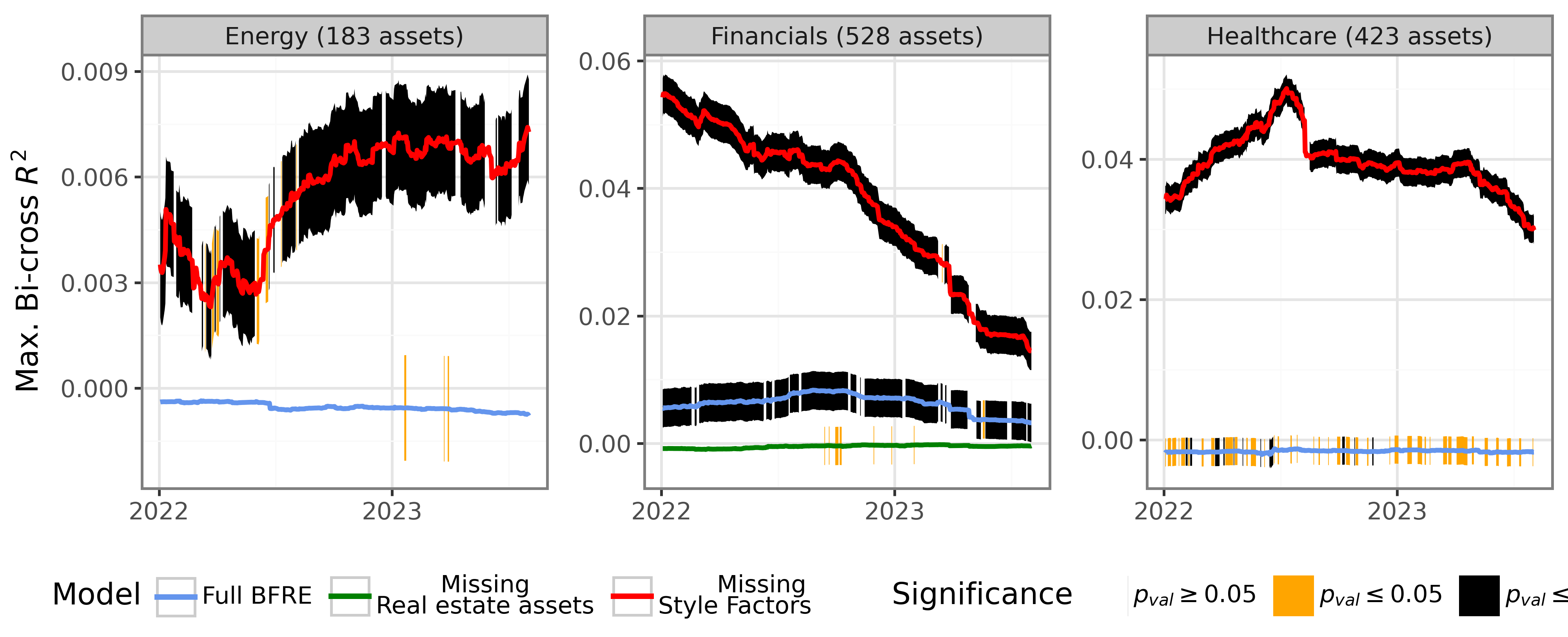}
    \caption{This figure plots the maximum bi-cross validation error statistic over time using a sliding window of $350$ observations, both for the BFRE model as well as the ablation study where we remove the twelve style factors from the BFRE model. For the financial sector, we also perform an analysis where we remove the stocks classified as real estate assets. The black and orange shading shows the significance of the mosaic p-value for this statistic. Occasionally, the p-value is significant even when the $R^2$ is negative, suggesting that the model does not fit perfectly but we do not know how to improve it (see Section \ref{subsec::improvement} for discussion).}
    \label{fig::r2_plot}
\end{figure}

\subsection{A missing factor among financial assets}

Previously, we consistently found the most significant evidence against the null in the financial sector. We now investigate this result. In short, we find evidence of unexplained correlations among certain real estate assets. 

In particular, when running (approximate) sparse PCA as per Eq. \ref{eq::sparse_pca} in the financial sector with $\ell=20$ stocks based on data up to January 2021, $100\%$ of the selected stocks have significant real estate exposure, and $85\%$ are classified primarily as ``real estate" assets in the BFRE model.\footnote{The BFRE model contains industry factors for real estate and real estate investment trusts (REITs), and $85\%$ of the selected assets are more exposed to one of these two industry factors than to any other industry.} 
To confirm this result, we recompute the maximum BCV $R^2$ statistic for the financial sector after removing real estate assets from the analysis. Figure \ref{fig::r2_plot} shows that our method is unable to improve the model for the financial sector after doing so. That said, this does not rule out the possibility that a more sophisticated analysis could detect additional unexplained correlations among financial stocks.

Figure \ref{fig::oos_corr} also plots an out-of-sample estimate of the correlation matrix of the (estimated) residuals of the top $\ell=10$ assets selected by the sparse PCA analysis. The assets are ordered by the sign of the estimated maximum eigenvector, and the plot confirms that assets with the same sign are positively correlated out-of-sample, and assets with different signs are negatively correlated. This result gives more confirmation that there are unexplained correlations among real estate assets. That said, the correlation patterns in Figure \ref{fig::oos_corr} are relatively sparse---for example, the residuals of MGIC Investment, the Radian Group, and the Essent group (which all provide mortgage insurance) are highly correlated with each other, but they are only weakly correlated with the other selected assets. This supports our finding that the existing model explains the most significant correlations among the assets.

\begin{figure}
    \centering
    \includegraphics[width=\linewidth]{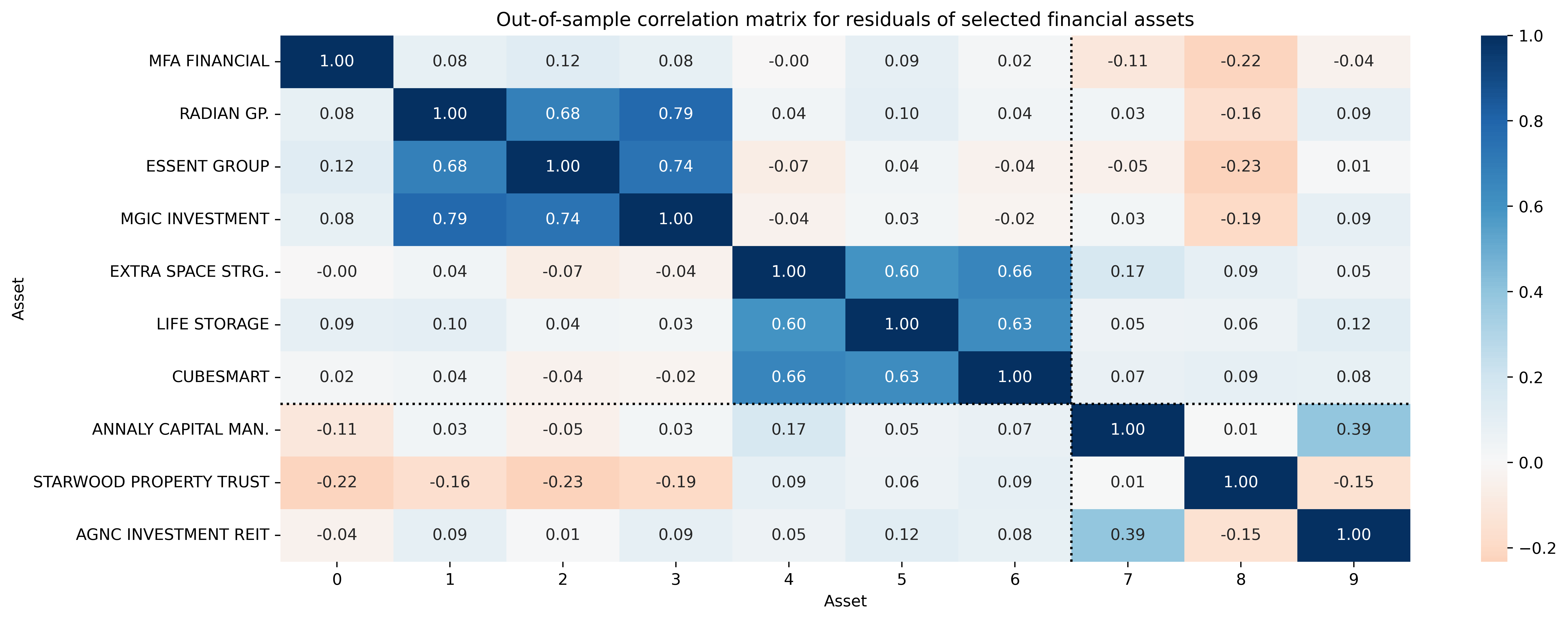}
    \caption{This figure shows the correlation matrix of the estimated residuals of the top ten assets selected by the sparse PCA analysis in Eq. \ref{eq::sparse_pca}. Note that the sparse PCA is performed on data up to January 2021, and the correlation estimates are computed out-of-sample using data after January 2021. The assets are shown in order of the value of the estimated eigenvector $\hat v$, and the dotted black line shows the point at which the eigenvector's entries switch from negative to positive. As expected, we see largely positive correlations on the block diagonal outlined by the black lines and largely negative correlations otherwise.}
    \label{fig::oos_corr}
\end{figure}

\section{Extensions}\label{sec::extensions}

\subsection{Adaptive choices of test statistic}\label{subsec::adaptive_test_stat}

How can one choose the most powerful test statistic from a set of candidates $S_1(\hat\epsilon), \dots, S_d(\hat \epsilon)$? We have already discussed a few answers to this question, such as (i) setting the final statistic $S = \max_{i=1}^d S_i(\hat\epsilon)$ to be the maximum of the candidates or (ii) using a single statistic which employs cross-validation to adaptively estimate nuisance parameters such as the sparsity level. However, this section explores another powerful and efficient technique.

As a concrete example, consider the \textit{quantile of maximum absolute correlations} (QMC) statistic. Recall that $\hmc_j = \max_{j' \ne j} |\hat C_{j,j'}|$ is the maximum absolute estimated correlation between the residuals of asset $j$ and another asset. The QMC statistic is the $\gamma$ empirical quantile of the estimated maximum absolute correlations:
\begin{equation}\label{eq::qmc}
    S_\gamma(\hat\epsilon) = Q_{\gamma}\left(\left\{\hmc_j : j \in [p]\right\}\right).
\end{equation}
Compared to the MMC statistic from Section \ref{sec::intro}, the QMC statistic is designed to increase power in sparse settings. For example, if only $1\%$ of the residuals are correlated, the MMC statistic will be corrupted by noise, whereas the value of the QMC statistic with $\gamma = 0.99$ will depend primarily on the $\hmc_j$ values for non-null assets. However, how can we choose $\gamma$ from candidates $\gamma_1, \dots, \gamma_d$? 

To choose among a set of $d$ candidate statistics $S_1(\hat\epsilon), \dots, S_d(\hat\epsilon)$, consider any function $f(\hat\epsilon, \tilde\epsilon^{(1)}, \dots, \tilde\epsilon^{(R)})$ which peeks at $\hat\epsilon$ and its mosaic permutations $\{\tilde\epsilon^{(r)}\}_{r=1}^R$ and aggregates evidence among all $d$ test statistics. We refer to $f : \R^{(R+1) \times T \times p} \to \R$ as a ``meta test-statistic."
E.g., $f$ could return the maximum approximate Z-statistic over all $d$ candidates:
\begin{equation}\label{eq::fadapt}
    f(\hat\epsilon, \tilde\epsilon^{(1)}, \dots, \tilde\epsilon^{(R)}) \stackrel{\mathrm{e.g.}}{=} \max_{i=1}^d \frac{S_i(\hat\epsilon) - \sum_{r=1}^R S_i(\tilde\epsilon^{(r)})}{\sqrt{\widehat{\var}(\{S_i(\tilde\epsilon^{(r)})\}_{r=1}^R)}}.
\end{equation}
To compute a p-value based on $f$, 
the key is that $\{\tilde\epsilon^{(r)}\}_{r=0}^R$ are exchangeable, where for notational convenience we set $\tilde\epsilon^{(0)} \defeq \hat\epsilon$. Therefore for any permutation $\pi : [R] \to [R]$, 
\begin{equation}
    f(\tilde\epsilon^{(0)}, \tilde\epsilon^{(1)}, \dots, \tilde\epsilon^{(R)}) \disteq f(\tilde\epsilon^{(\pi(0))}, \tilde\epsilon^{(\pi(1))}, \dots, \tilde\epsilon^{(\pi(R))}).
\end{equation}
Thus, we can compute a p-value based on any meta test-statistic by randomly permuting the order of $\{\tilde\epsilon^{(r)}\}_{r=0}^R$ and checking if this decreases the value of $f$. Algorithm \ref{alg::adaptive_test_stat} formally describes this procedure.

\begin{algorithm}[!h]
\caption{Adaptive meta-test statistic}\label{alg::adaptive_test_stat}
\textbf{Input}: Returns $\bY \in \R^{T \times p}$, exposures $L_t \in \R^{p \times k}$ for $t \in [T]$, tiles $\{(B_m, G_m)\}_{m=1}^M$ and a meta test-statistic $f : \R^{(R+1) \times T \times p} \to \R$.
\begin{steps}[topsep=0pt, itemsep=0.5pt, leftmargin=*]
    \item Construct the mosaic residual estimate $\hat\epsilon \in \R^{T \times p}$ and its permuted variants $\tilde\epsilon^{(r)} \in \R^{T \times p}$ for $r=1,\dots,R$, as described in Algorithm \mosaicpermalgnum. Set $\tilde\epsilon^{(0)} = \hat\epsilon$.
    \item Compute the original meta test-statistic $f(\tilde\epsilon^{(0)}, \dots, \tilde\epsilon^{(R)})$.
    \item Sample uniformly random permutations $\pi_1, \dots, \pi_K : \{0, \dots, R\} \to \{0, \dots, R\}$.
    \item Compute the final adaptive p-value \begin{equation}\label{eq::padapt}
    p_{\mathrm{adaptive}} = \frac{1 + \sum_{\ell=1}^{K} \I\left(f(\tilde\epsilon^{(0)}, \tilde\epsilon^{(1)}, \dots, \tilde\epsilon^{(R)}) \le f(\tilde\epsilon^{(\pi_\ell(0))}, \tilde\epsilon^{(\pi_\ell(1))}, \dots, \tilde\epsilon^{(\pi_\ell(R))})\right)}{1 + K}.
\end{equation}
\end{steps}
\end{algorithm}

\begin{corollary}\label{cor::adaptive_test_stat} $p_{\mathrm{adaptive}}$ is a valid p-value under the assumptions of Theorem \ref{thm::main}.
\end{corollary}

Algorithm \ref{alg::adaptive_test_stat} uses \textit{two} distinct layers of permutations: one set of mosaic permutations to construct $\{\tilde\epsilon^{(r)}\}_{r=1}^R$, and a second set of simple random permutations $\pi_1, \dots, \pi_K$ to compute an adaptive p-value based on $f$. The second layer of permutations adds conceptual complexity, but it yields a highly computationally efficient procedure compared to one which (e.g.) repeatedly cross-validates an expensive machine learning algorithm. 

\subsection{Adaptively choosing the tiling}\label{subsec::adpativetiling}

Although Section \ref{subsec::defaulttiling} gives a good default choice of tiling, another option is to \textit{learn} a good choice of tiles that ``separate" assets whose idiosyncratic returns are correlated. However, in general, if the tiling is chosen using $\bY$, $\hat \epsilon$ will not necessarily be invariant to any permutations under the null because of the dependence between $\{(B_m, G_m)\}_{m=1}^M$ and $\bY$. In other words, naive ``double dipping" leads to inflated false positives. 

However, we can sequentially choose the $m$th tile $(B_m, G_m)$ based on the estimated residuals from the previous $m-1$ tiles as long as our choice of $(B_m, G_m)$ does not depend on the \emph{order} of the rows within each of the previous $m-1$ tiles. Precisely, suppose that we can write $B_m$, $G_m$ as functions $b_m, g_m$ of the previous tiles as well as auxiliary randomness $U_m \iid \Unif(0,1)$:  
\begin{equation}\label{eq::bmgm}
    B_m = b_m(\hat\epsilon_{(1)}, \dots, \hat\epsilon_{(m-1)}, U_m) \text{ and } G_m = g_m(\hat\epsilon_{(1)}, \dots, \hat\epsilon_{(m-1)}, U_m).
\end{equation}
If $b_m$ and $g_m$ are invariant to permutations of the previous tiles, then $\pval$ will be a valid p-value, as stated by the following lemma. Intuitively, the proof of the lemma follows from the fact that Eq. \ref{eq::full_disteq} still holds after conditioning on the tiles $\{(B_m, G_m)\}_{m=1}^M$ (see Appendix \ref{appendix::proofs} for a proof).

\begin{lemma}\label{lem::adaptive_tile} Suppose for each $m=1,\dots, M$, $(B_m, G_m)$ does not depend on the order of the rows within the first $m-1$ tiles. Formally, for any permutation matrices $P_j \in \R^{|B_j| \times |B_j|}$, we assume that
\begin{equation*}
    b_m(\hat\epsilon_{(1)}, \dots, \hat\epsilon_{(m-1)}, U_m) 
    =
    b_m(P_1 \hat\epsilon_{(1)}, \dots, P_{m-1} \hat\epsilon_{(m-1)}, U_m),
\end{equation*}
and the same holds when replacing $b_m$ with $g_m$. Suppose also that Assumption \ref{assump::batched_exposure} holds and, for simplicity, that $\{\epsilon_{t,j}\}_{t=1}^T \iid P_j$ is i.i.d. for each $j \in [p]$. Then $p_{\mathrm{val}}$ is still a valid p-value testing $\mcH_0$ as in Theorem \ref{thm::main}.
\end{lemma}

\begin{remark} In Appendix \ref{appendix::proofs}, we relax the i.i.d. assumption in Lemma \ref{lem::adaptive_tile} to a local exchangeability assumption (see Remark \ref{rem::localexch_support}).
For simplicity, we defer this extension to Appendix \ref{appendix::proofs}. 
\end{remark}

Lemma \ref{lem::adaptive_tile} allows us to use many different methods to adaptively choose the tiling. For example, our default non-adaptive choice of tiling involved randomly partitioning the assets into $D$ groups $[p] = G_{i,1} \cup \dots \cup G_{i,D}$ for batch $i \in [I]$ of the observations. Instead, one could make an adaptive choice satisfying the permutation-invariance constraint from Lemma \ref{lem::adaptive_tile} as follows. For $i=1$, we choose the groups randomly as before. Then, sequentially for $i \ge 2$, we let $\hat\Sigma^{(i)}$ be the empirical covariance estimator which (a) only uses information from the first $i-1$ batches and (b) only uses information from \textit{within} tiles (and not between tiles). Precisely, fix any pair of assets $j, j' \in [p]$. Let $A = \left\{t \in [T] : (j, j') \in G_{i_0,d} \text{ for some } i_0 \le i, d \in [D] \right\}$ be the set of time points in the first $i-1$ batches where assets $j,j'$ are in the same group. (Note $A$ depends on $i, j, j'$ but for simplicity, we suppress this dependence). The within-tile covariance estimate is defined as
\begin{equation*}
    \hat\Sigma^{(i)}_{j,j'} = \frac{1}{|A|} \sum_{t \in A} \hat\epsilon_{t,j} \hat\epsilon_{t,j'} - \left(\frac{1}{|A|} \sum_{t \in A} \hat\epsilon_{t,j} \right) \left(\frac{1}{|A|} \sum_{t \in A} \hat\epsilon_{t,j'}\right).
\end{equation*}
It is easy to see that $\hat\Sigma^{(i)}$ does not depend on the ordering of the rows of any of the tiles. Thus, we can choose the groups for the $i$th batch based on $\hat\Sigma^{(i)}$ while preserving validity. In particular, we suggest choosing partitions which approximately solve the following optimization problem:
\begin{equation}\label{eq::estg}
    G_{i,1}, \dots, G_{i,D} = \argmax_{G_{i,1}, \dots, G_{i,D} \text{ partition } [p]} \sum_{j,j' \in [p]} |\hat \Sigma^{(i)}_{j,j'}| \cdot \I(j, j' \text{ are in different groups}).
\end{equation}
In other words, for each $i$, we choose $G_{i,1}, \dots, G_{i,D}$ to maximize the sum of the absolute estimated correlations between residuals that are not in the same group. Although exactly solving this optimization problem is computationally prohibitive, we can solve it approximately using a greedy randomized algorithm outlined in Appendix \ref{appendix::greedy_adaptive_tiles}. 

\subsection{Allowing the exposures to change with each observation}\label{subsec::nonconst_lt}

Motivated by our real applications, our analysis so far assumes that the exposures $L_t \in \R^{p \times k}$ are constant within tiles. However, if $L_t$ changes with every observation, a fix is to replace $L_t$ with an \textit{augmented} exposure matrix $L_t\opt$:
\begin{equation}\label{eq::aug_exposures}
    L_t\opt \defeq \begin{cases} 
    \begin{bmatrix} L_t & L_{t+1} \end{bmatrix} & t \text{ is odd} \\
    \begin{bmatrix} L_{t-1} & L_{t} \end{bmatrix} & t \text{ is even } \end{cases} \in \R^{p \times 2k}.
\end{equation}
By construction, $L_t\opt$ only changes every two observations. E.g., for the first two time points, $L_1\opt = L_2\opt = \begin{bmatrix} L_1 & L_2 \end{bmatrix} \in \R^{p \times 2k}.$ Furthermore, since $L_t$ is a submatrix of $L_t\opt$, if the null holds for the original model $Y_t = L_t X_t + \epsilon_t$, it also holds for the augmented model $Y_t = L_t\opt X_t\opt + \epsilon_t$, since we can set $X_t\opt \in \R^{2k}$ to equal $(X_t, 0)$ for even $t$ and $(0, X_t)$ for odd $t$. Thus, after augmenting the exposures, we can apply the mosaic permutation test (all tiles will contain exactly two observations to ensure $L_t\opt$ is constant within tiles). However, the cost is that we must estimate twice as many nuisance parameters when estimating $\hat\epsilon$.

\section{Do the mosaic residual estimates cause a loss of power?}\label{sec::sims}

Our method requires the test statistic $S(\hat\epsilon)$ to be a function of \textit{mosaic} residual estimates instead of a function $S(\hat\epsilon\ols)$ of OLS residual estimates. 
We hope that the mosaic statistic is a good proxy for the OLS statistic (which indeed seems to be the case in Figure \ref{fig::motivation_soln}), so this section analyzes via simulations whether the mosaic test has lower power than an oracle test based on the OLS statistic. Our simulations also show the effectiveness of the adaptive test statistic introduced in Section \ref{subsec::adaptive_test_stat}.

We conduct semisynthetic simulations where the exposures $L_t$ are constant over time and equal to the BFRE exposures for the financial sector on \stardate. We sample the factor returns $X_{tk}$ as i.i.d. $t_4$ variables. The residuals satisfy $\epsilon_t = \gamma_t + Z_t v$ for $\gamma_{t,j} \iid t_4, Z_t \iid t_4$ and $v \in \R^p$. In words, the residuals are i.i.d. $t_4$ variables plus a ``missing" factor component based on an extra factor return $Z_t \in \R$ and corresponding exposures $v \in \R^p$. We let $v$ have $\ceil{s_0 p}$ nonzero coordinates chosen uniformly at random, and the nonzero values equal $\frac{\rho}{\sqrt{\ceil{s_0 p}}}$---thus, $\rho$ measures the signal size and $s_0$ measures the sparsity. All simulations use $T=50$ observations. We compare three methods: 
\begin{enumerate}[itemsep=0.5pt, topsep=0pt, leftmargin=*]
    \item First, we apply the mosaic permutation test (MPT) with the default choice of tiling and the adaptive quantile maximum correlation (QMC) statistic from Eq. \ref{eq::qmc}. Since we do not know the optimal quantile $\gamma$ a priori, we use Eqs. \ref{eq::fadapt} and \ref{eq::padapt} to compute an adaptive p-value which aggregates evidence across the test statistics using $\gamma \in \{0.01, 0.1, 0.25, 0.5, 0.75, 0.9, 0.99\}$.
    \item Second, we compare this to the performance of the MPT with an ``oracle" QMC statistic, where we pick the single value of $\gamma$ which maximizes power---this is an oracle because, in practice, the optimal value of $\gamma$ would be unknown. 
    \item Lastly, we also compute the oracle QMC statistic applied to the OLS residuals $\hat\epsilon\ols$. We check the significance of the OLS oracle QMC statistic by comparing it to its true null distribution (with $\rho = 0$). This is a ``doubly oracle" test statistic since, in a real data analysis, we would not know the optimal choice of $\gamma$, nor would we know the OLS statistic's null distribution. That said, comparing to this OLS statistic will help us understand whether the mosaic residual estimates cause a loss in power. 
\end{enumerate}

Figure \ref{fig::mainsims} shows the results. For various sparsities $s_0$ and signal sizes $\rho$, the MPT does not lose much power compared to the oracle tests, suggesting that the adaptive QMC statistic effectively adapts to the unknown sparsity level. Furthermore, the MPT oracle and OLS double oracle---which use the same test statistic but are applied to different residual estimates---have similar power. Indeed, the average power difference is $3\%$, and the maximum power difference is $10\%$. This result should not be too surprising, since $\hat\epsilon\ols$ and $\hat\epsilon$ are estimating the same residuals $\epsilon$, and thus the OLS and mosaic statistic should be highly correlated. Indeed, in our application (Figure \ref{fig::motivation_soln}), in all three sectors, the mosaic statistics are empirically $\ge 85\%$ correlated with the OLS statistics. Thus, in this simulation, the mosaic test is competitive with an oracle method based on OLS residuals.

\begin{figure}
    \centering
    \includegraphics[width=\linewidth]{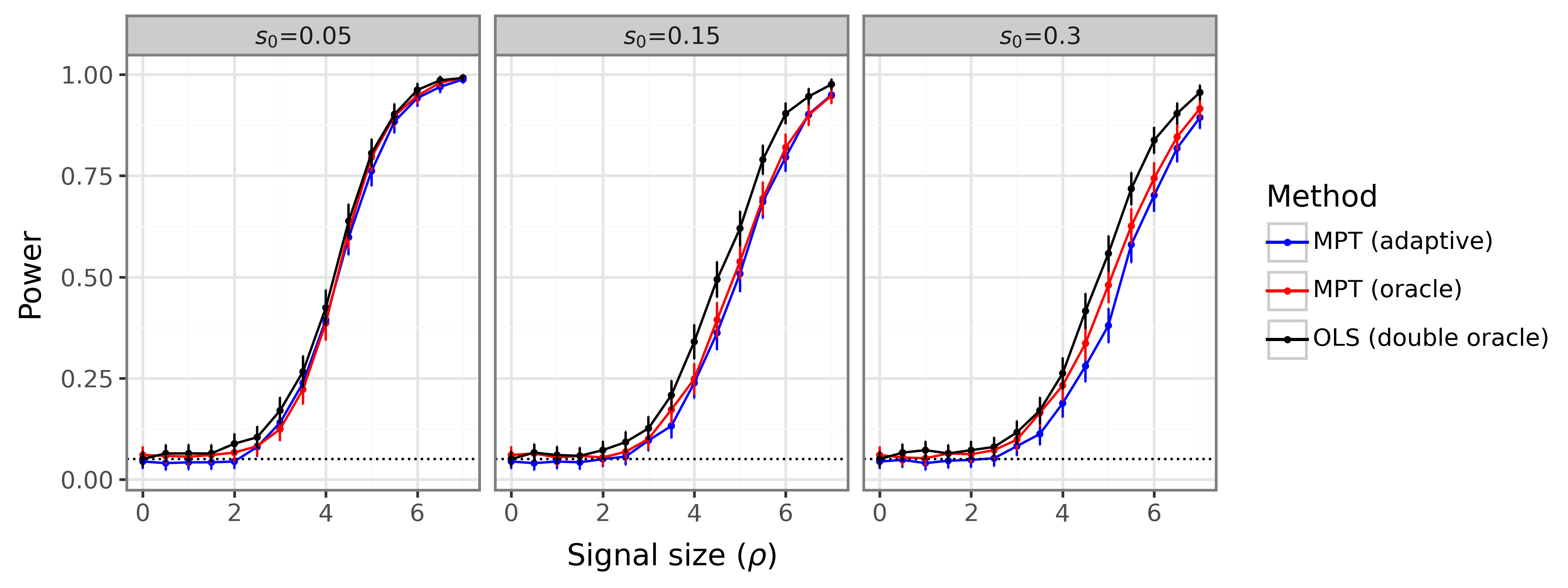}
    \caption{This figure shows the power of the mosaic permutation test with an adaptive QMC statistic as well as the power of the two oracles from Section \ref{sec::sims}. It shows that (1) the adaptive QMC statistic from Section \ref{subsec::adaptive_test_stat} effectively adapts to the unknown sparsity of the alternative, and (2) the MPT does not lose much power compared to an oracle procedure using the OLS residuals $\hat\epsilon\ols$ in place of the mosaic residuals $\hat\epsilon$. For example, the average power difference between the MPT oracle and OLS double oracle (which use the same test statistic but applied to different residual estimates) is $3\%$, and the maximum power difference is $10\%$.
    The dotted black line shows the nominal level $\alpha = 0.05$. All methods control the false positive rate when $\rho = 0$ and $\mcH_0$ holds.} 
    \label{fig::mainsims}
\end{figure}

\section{Discussion}\label{sec::discussion}

This paper introduces the mosaic permutation test, an exact and nonparametric goodness-of-fit test for factor models with known exposures. In an empirical application to the BlackRock Fundamental Equity Risk model, we demonstrate how to use the mosaic permutation test to diagnose financial factor models. Additionally, our simulations and theory show the power and flexibility of the mosaic permutation test, which can be used in combination with a wide variety of test statistics to quickly detect unexplained correlations among variables. Lastly, although this paper focuses on applications to financial factor models, our methods can also be applied to test the goodness-of-fit of pre-existing factor models in any domain, including psychology \citep[e.g.,][]{mccrae1992fivefactorpsych} and genetics \citep[e.g.,][]{gain2021geneticfactormodels}. That said, our work leaves open several possible directions for future research.
\begin{itemize}[leftmargin=*]
    \item \underline{Confidence intervals}. It would be interesting to produce confidence intervals that quantify \textit{how severely} the null is violated. Indeed, this might help analysts understand the economic significance of any rejections made by our methods.
    \item \underline{Anytime-valid tests}. If one sequentially produces many mosaic p-values over time, one will eventually obtain a false positive by chance \citep{ramdas2023gametheoretic}. This limitation should not affect our empirical findings, since the p-values in (e.g.) Figure \ref{fig::mmc_pval} are small enough to remain significant even after an appropriate multiplicity correction. Nonetheless, producing an anytime-valid test may be worthwhile.
    \item \underline{Factor models with known factor returns}. This paper analyzes factor models with known exposures $L_t$. However, in some cases, it may be more realistic to assume the factor returns $X_t$ are known. In future work, we plan to extend the mosaic permutation test to apply to this new setting (among others). 
    \item \underline{Using regularization}: Our methods currently require the use of \textit{unregularized} OLS regressions within each tile to estimate the residuals (see Remark \ref{rem::regularization}). However, to increase power, it might be valuable to develop methods that can use regularization.
    \item \underline{Robustness}: 
    It might be valuable to develop tests that are robust to slight inaccuracies in the exposures $L_t$. Indeed, this could also help relax the assumption that $L_t$ is locally constant (Assumption \ref{assump::batched_exposure}), since small within-tile changes in $L_t$ could be viewed as small ``inaccuracies." Similarly, it would be useful to develop theory quantifying the robustness of the existing test, i.e., by bounding the excess error in some interpretable way.
\end{itemize}

\section{Code and data availability}

We implemented our methods in the python package \texttt{mosaicperm}. All other code used in the paper is available at \url{https://github.com/amspector100/mosaic_factor_paper/}. Although we are not able to make the BFRE model data available, we have provided a publicly available sample dataset that allows one to obtain qualitatively similar results (see the GitHub repository for more details).

\section{Acknowledgements}

The authors would like to thank John Cherian, Kevin Guo, Guido Imbens, Lihua Lei, and Bobby Luo for valuable suggestions. A.S. was partially supported by the Two Sigma Graduate Fellowship Fund, the Citadel GQS PhD Fellowship, and a Graduate Research Fellowship from the National Science Foundation. R.F.B. was supported by the Office of Naval Research via grant N00014-20-1-2337 and by the National Science Foundation via grant DMS-2023109. E.J.C. was supported by the Office of Naval Research grant N00014-20-1-2157, the National Science Foundation grant DMS-2032014, and the Simons Foundation under award 814641. We disclose that R. Kahn is an employee of BlackRock, and E. Candes and T. Hastie are scientific advisors at BlackRock (see https://www.blackrock.com/corporate/ai). This information is relevant since the paper analyzes a BlackRock factor model. However, the authors certify that they have no particular interest in any of the results in the paper. We also disclose that a BlackRock employee reviewed this paper before its circulation. They approved the manuscript with no modifications.

\bibliography{ref}
\bibliographystyle{apalike}
\appendix

\section{Proofs}\label{appendix::proofs}

\subsection{Proof of Theorem \ref{thm::main} and Corollary \ref{cor::adaptive_test_stat}}\label{subsec::mainproofs}

We now prove Theorem \ref{thm::main} and Corollary \ref{cor::adaptive_test_stat}. (Theorem \ref{thm::validitybasic} is a case of Theorem \ref{thm::main} with $M=2$ tiles.) Then, we consider the case where some data are missing (Corollary \ref{cor::missingdata}).
 
\begingroup
\def\thetheorem{\ref{thm::main}} 
\begin{theorem} Suppose Assumptions \ref{assump::batched_exposure}-\ref{assump::localexch}. Then under $\mcH_0$, Eq. \ref{eq::hatepspval2} defines a valid p-value satisfying $\P(\pval \le \alpha) \le \alpha$ for any $\alpha \in (0,1)$. 
\begin{proof} The proof is in three steps. The first step reviews two useful consequences of Assumption \ref{assump::localexch} (namely Equations \ref{eq::mainstepone} and \ref{eq::mainsteptwo}). The second step uses these results to show that the mosaic estimate $\hat\epsilon$ and its permuted variants $\tilde\epsilon^{(1)}, \dots, \tilde\epsilon^{(R)}$ are jointly exchangeable. The final step shows the validity of the resulting p-value. These steps use standard statistical results about exchangeability, but we include them all for completeness.

\underline{Step 1}: Recall that $\epsilon_{(m)} \defeq \epsilon_{B_m, G_m} \in \R^{|B_m| \times |G_m|}$ denotes the residuals in the $m$th tile, for $m \in [M]$. Observe that under $\mcH_0$ and Assumption \ref{assump::localexch}, for any permutation matrices $P_1 \in \R^{|B_1| \times |B_1|}, \dots, P_M \in \R^{|B_M| \times |B_M|}$, we have the distributional equality:
\begin{equation}\label{eq::mainstepone}
    (P_1 \epsilon_{(1)}, \dots, P_M \epsilon_{(M)}) \disteq (\epsilon_{(1)}, \dots, \epsilon_{(M)}).
\end{equation}
The above equation holds because each column (asset) in $\epsilon$ is separately permuted with a permutation that only swaps entries \textit{within} tiles (not between tiles). Local exchangeability (Assumption \ref{assump::localexch}) guarantees that this does not change the \textit{marginal} distribution of each column of $\epsilon$, and the null $\mcH_0$ guarantees that all columns of $\epsilon$ are independent.

Notably, this implies the following result. Let $P_m^{(r)} \in \R^{|B_m| \times |B_m|}$ denote the randomly sampled permutation matrix for tile $m$ for the  $r = 1, \dots, R$ randomizations in Algorithm \mosaicpermalgnum, and let $P_m^{(0)} = I_{|B_m|}$ denote the identity matrix. Then the following holds:

\begin{equation}\label{eq::mainsteptwo}
    \left[(P_1^{(r)} \epsilon_{(1)}, \dots, P_M^{(r)} \epsilon_{(M)})\right]_{r=0}^R \text{ are exchangeable.}
\end{equation}

This is a standard consequence of Equation \ref{eq::mainstepone}, so we defer the proof to Lemma \ref{lem::technical_exch_result}.

\underline{Step 2}: Recall from Equation \ref{eq::hateps_tile_def} that the $m$th estimated tile satisfies $\hat \epsilon_{(m)} \defeq \epsilon_{(m)} H_m$ for a deterministic projection matrix $H_m \in \R^{|G_m| \times |G_m|}$. This notation implicitly uses the assumption that the exposures do not change within tiles (Assumption \ref{assump::batched_exposure}), as otherwise, $H_m$ might change for different rows of $\hat\epsilon_{(m)}$. Since $H_m$ is a deterministic matrix for $m \in [M]$, Equation \ref{eq::mainsteptwo} immediately implies
\begin{equation}
    \left[(P_1^{(r)} \epsilon_{(1)} H_1, \dots, P_M^{(r)} \epsilon_{(M)} H_M)\right]_{r=0}^R \text{ are exchangeable.}
\end{equation}
Recall that by definition, $P_m^{(0)}$ is the identity matrix, so $\hat\epsilon$ is simply a deterministic concatenation of $(P_1^{(0)} \epsilon_{(1)} H_1, \dots, P_M^{(0)} \epsilon_{(M)} H_M)$. Similarly, by definition, $\tilde\epsilon^{(r)}$ is just the appropriate concatenation of $(P_1^{(r)} \epsilon_{(1)} H_1, \dots, P_M^{(r)} \epsilon_{(M)} H_M)$. Thus, the previous equation implies
\begin{equation}\label{eq::tilde_eps_exch}
    (\hat \epsilon, \tilde\epsilon^{(1)}, \dots, \tilde\epsilon^{(R)}) \text{ are exchangeable.}
\end{equation}

\underline{Step 3}: We now show that $p_{\mathrm{val}}$ is a valid p-value using Step 2. In particular, since $S : \R^{T \times p} \to \R$ is a deterministic function, we know that $(S(\hat \epsilon), S(\tilde\epsilon^{(1)}), \dots, S(\tilde\epsilon^{(R)})) \text{ are exchangeable.}$ This implies that if $\tau$ denotes the rank of $S(\hat\epsilon)$ among $S(\hat\epsilon), S(\tilde\epsilon^{(1)}), \dots, S(\tilde\epsilon^{(R)})$ where ties are broken uniformly at random and smaller ranks denote larger values, then $\tau \sim \Unif(\{1, \dots, R+1\})$. Note, however, that 
\begin{equation*}
    p_{\mathrm{val}} = \frac{\sum_{r=1}^{R+1} \I(S(\hat\epsilon) \le S(\tilde\epsilon^{(r)})) + 1}{R+1} \ge \frac{\tau}{R+1},
\end{equation*}
where the deterministic equality follows because $\tau$ breaks ties uniformly at random but $p_{\mathrm{val}}$ is defined to always break ties conservatively. Thus, $p_{\mathrm{val}}$ stochastically dominates $\frac{\tau}{R+1}$, proving that $\P(\pval \le \alpha) \le \P(\tau \le (R+1) \alpha) \le \alpha$. This also proves that if there are no ties with probability $1$, then $p_{\mathrm{val}} = \frac{\tau}{R+1} \sim \Unif\left(\{\frac{1}{R+1}, \dots, 1 \}\right)$.
\end{proof}
\end{theorem}
\endgroup

\begingroup
\def\thecorollary{\ref{cor::adaptive_test_stat}}
\begin{corollary} $p_{\mathrm{adaptive}}$ is a valid p-value under the assumptions of Theorem \ref{thm::main}.
\end{corollary}
\begin{proof} Recall from Eq. \ref{eq::tilde_eps_exch} that $\{\tilde\epsilon^{(r)}\}_{r=0}^R$ are exchangeable, where $\tilde\epsilon^{(0)} \defeq \hat\epsilon$. 
Also, recall that $p_{\mathrm{adaptive}}$ is defined as follows. Let $\pi_k : \{0, \dots, R\} \to \{0, \dots, R\}$ be uniformly random permutations for $1 \le k \le K$ and $\pi_0$ be the identity mapping. For a deterministic ``meta test-statistic" $f : \R^{(R+1) \times T \times p} \to \R$, define $T_k = f(\tilde\epsilon^{(\pi_k(0))}, \dots, \tilde\epsilon^{(\pi_k(R))})$ for $0 \le k \le K$. Then
\begin{equation*}
    p_{\mathrm{adaptive}} \defeq \frac{1 + \sum_{k=1}^K \I(T_0 \le T_k)}{K+1}.
\end{equation*}
Using the same logic as Step 3 in the proof of Theorem \ref{thm::main}, it suffices to show that $\{T_k\}_{0 \le k \le K}$ are exchangeable. However, since $\{\tilde\epsilon^{(r)}\}_{r=0}^R$ are exchangeable and $\pi_1, \dots, \pi_K$ are uniformly random permutations, $\{T_k\}_{0 \le k \le K}$ are exchangeable as well. (This is a standard statistical argument, although for completeness, we prove this result in Lemma \ref{lem::simple_exch_result}.)
\end{proof}
\endgroup
\addtocounter{theorem}{-1}

\begin{remark}\label{rem::missingdata}
The proof of Theorem \ref{thm::main} requires that the tiles $\{B_m \times G_m\}_{m=1}^M$ are disjoint, but nowhere does it require that they fully partition $[T] \times [p]$, i.e., $\bigcup_{m=1}^M B_m \times G_m$ need not equal $[T] \times [p]$. This suggests a simple strategy to deal with missing data, which is to choose the tiles so that no tile contains any missing data (this means that some observations $Y_{t,j}$ will not be in any tile). Then, we proceed as usual: we define the mosaic residuals as the concatenation of the residuals in each tile,
$$\hat\epsilon = \left(\hat\epsilon_{(1)}, \dots, \hat\epsilon_{(M)} \right) \in \R^{|B_1| \times |G_1|} \times \dots \times \R^{|B_M| \times |G_M|},$$
and the permuted residuals are defined as $\tilde\epsilon^{(r)} = (P_1^{(r)} \hat\epsilon_{(1)}, \dots, P_M^{(r)} \hat\epsilon_{(M)})$, where $P_m^{(r)} \in \{0,1\}^{|B_m| \times |B_m|}$ are uniformly random permutation matrices. Although $\hat\epsilon$ and $\tilde\epsilon^{(r)}$ are no longer rectangular matrices, we can still define the p-value as in Eq. \ref{eq::hatepspval2}. Namely, for any test statistic $S : \prod_{m=1}^M \R^{|B_m| \times |G_m|} \to \R$, the p-value is
\begin{equation}\label{eq::missingdata_pval}
    p_{\val} = \frac{1 + \sum_{r=1}^R \I(S(\hat\epsilon) \le S(\tilde\epsilon^{(r)}))}{R+1}.
\end{equation}
The following corollary proves the validity of this p-value.
\end{remark}

\begin{corollary}\label{cor::missingdata} For each $j \in [p], t \in [T]$, let $A_{t,j} \in \{0,1\}$ be the indicator of whether $Y_{t,j}$ is observed. Suppose the following:
\begin{enumerate}[noitemsep, topsep=0pt]
    \item The tiles $\{B_m, G_m\}_{m=1}^M$ are chosen such that (i) Assumption \ref{assump::batched_exposure} holds and (ii) in each tile, no data are missing. That is, for all $m \in [M], t \in B_m, j \in G_m$, $A_{t,j} = 1$. 
    \item The tiles $\{B_m, G_m\}_{m=1}^M$ are chosen as a function of the missing data pattern $\{A_{t,j}\}_{t \in [T], j\in [p]}$ and an independent uniform $U \sim \Unif(0,1)$. In other words, one cannot use $Y$ to construct the tiles.  
    \item Assumption \ref{assump::localexch} (local exchangeability) holds conditional on the missing data pattern $\{A_{t,j}\}_{t \in [T], j\in [p]}$.
\end{enumerate}
Then if $p_{\val}$ is defined using Eq. \ref{eq::missingdata_pval}, $\P(p_{\val} \le \alpha) \le \alpha$ for any $\alpha \in (0,1)$ under $\mcH_0$.
\begin{proof} 
To prove the result, condition on the missing data pattern $\{A_{t,j}\}_{t \in [T], j\in [p]}$ and the external uniform $U \sim \Unif(0,1)$ which may be used to construct randomized tiles. Conditional on these quantities, the tiles $\{B_m, G_m\}_{m=1}^M$ are deterministic, Assumptions \ref{assump::batched_exposure}-\ref{assump::localexch} (locally constant exposures and local exchangeability) both hold by hypothesis, and no tile contains missing data. In other words, the assumptions of Theorem \ref{thm::main} still hold---the only difference is that now, the union of the tiles $\prod_{m=1}^M B_m \times G_m$ is no longer equal to $[T] \times [p]$. However, by the previous remark, this assumption is not needed in the proof of Theorem \ref{thm::main}. This completes the proof.
\end{proof}
\end{corollary}

\subsection{Properties of the approximate Z-statistic}\label{appendix::apprx_Z}

We now show that the approximate Z-statistic in Eq. \ref{eq::apprx_Z} has mean zero and unit variance under the null. This result is a simple consequence of exchangeability, but we state it for completeness.

\begin{lemma}\label{lem::apprx_Z} Let $\hat\epsilon \in \R^{T \times p}$ be mosaic residuals and $\tilde\epsilon^{(1)}, \dots, \tilde\epsilon^{(R)} \in \R^{T \times p}$ be $R$ permuted variants. For any test statistic $S : \R^{T \times p} \to \R$, define
\begin{equation}\label{eq::apprx_Z_appendix}
    Z\apprx = \frac{S(\hat\epsilon) - \bar S}{\sqrt{\frac{1}{R+1} \sum_{r=0}^R (S(\tilde\epsilon^{(r)}) - \bar S)^2}} \text{ for } \bar S = \sum_{r=0}^{R+1} S(\tilde\epsilon^{(r)}),
\end{equation}
where $\tilde\epsilon^{(0)} \defeq \hat \epsilon$. Then under $\mcH_0$ and Assumptions \ref{assump::batched_exposure}-\ref{assump::localexch}, $\E[Z\apprx] = 0$ and $\var(Z\apprx) = 1$.
\begin{proof} As notation, let $Z\apprx^{(r)} = \frac{S(\tilde\epsilon^{(r)}) - \bar S}{\sqrt{\frac{1}{R+1} \sum_{r=0}^R (S(\tilde\epsilon^{(r)}) - \bar S)^2}}$, for $r=0, \dots, R$, so $Z\apprx = Z\apprx^{(0)}$. By construction, 
\begin{equation*}
    \sum_{r=0}^R Z\apprx^{(r)} = 0 \text{ and } \sum_{r=0}^R (Z\apprx^{(r)})^2 = \frac{\sum_{r=0}^R (S(\tilde\epsilon^{(r)}) - \bar S)^2}{\frac{1}{R+1} \sum_{r=0}^R (S(\tilde\epsilon^{(r)}) - \bar S)^2} = R + 1.
\end{equation*}
By exchangeability of $\{\tilde\epsilon^{(r)}\}_{r=0}^R$, we have that $\{Z\apprx^{(r)}\}_{r=0}^R$ are exchangeable as well. Thus, taking expectations, the previous equation implies that
\begin{equation*}
    \sum_{r=0}^R \E[Z\apprx^{(r)}] = (R+1) \E[Z\apprx^{(0)}] = 0 \text{ and } 
    \sum_{r=0}^R \E[(Z\apprx^{(r)})^2] = (R+1) \E[(Z\apprx^{(0)})^2] = R + 1.
\end{equation*}
Dividing by $R+1$ on both sides yields the result.
\end{proof}
\end{lemma}

\subsection{Proof of Lemma \ref{lem::adaptive_tile}}

We now prove a slightly more general version of Lemma \ref{lem::adaptive_tile}. In particular, Lemma \ref{lem::adaptive_tile} assumed for simplicity that for each asset, the true residuals $\{\epsilon_{t,j}\}_{t=1}^T \iid P_j$ were i.i.d. This can be relaxed to a type of local exchangeability---however, we cannot directly assume Assumption \ref{assump::localexch}, because Assumption \ref{assump::localexch} defines local exchangeability with respect to a \textit{fixed} tiling, and in Lemma \ref{lem::adaptive_tile}, the tiling $\{(B_m, G_m)\}_{m=1}^M$ is random. Instead, we must assume the following:

\begin{assumption}\label{assump::localexch_support} Local exchangeability (Assumption \ref{assump::localexch}) holds with respect to any fixed tiling $\{(\beta_m, \gamma_m)\}_{m=1}^M$ in the support of the learned (random) tiling $\{(B_m, G_m\}_{m=1}^M$. I.e., Assumption \ref{assump::localexch} holds for any fixed tiling satisfying $\P(\{(B_m, G_m)\}_{m=1}^M = \{(\beta_m, \gamma_m)\}_{m=1}^M) > 0$.
\end{assumption}

\begin{remark} Assumption \ref{assump::localexch_support} always holds if the residuals $\{\epsilon_{t,j}\}_{t=1}^T \iid P_j$ are i.i.d. for each asset.
\end{remark}

\begin{remark}\label{rem::localexch_support}
In practice, Assumption \ref{assump::localexch_support} is not much stronger than Assumption \ref{assump::localexch}. In particular, there is not much benefit to adaptively choosing the batches $B_1, \dots, B_M \subset [T]$---the main benefit is to choose the groups of assets $G_1, \dots, G_M \subset [p]$ adaptively to ``separate" assets with highly correlated residuals. Thus, in Section \ref{subsec::adpativetiling}, we suggested choosing $B_1, \dots, B_M$ so that each $B_m$ is a member of a prespecified partition $\beta_1, \dots, \beta_{M\opt} \subset [T]$ of $[T]$, where $\beta_1 = \{1, \dots, 10\}, \beta_2 = \{11, \dots, 20\}$, etc. In this case, Assumption \ref{assump::localexch_support} reduces to the regular local exchangeability assumption (Assumption \ref{assump::localexch}) with respect to the tiling $(\{\beta_m, [p]\}_{m=1}^{M\opt}$.     
\end{remark}

Having stated Assumption \ref{assump::localexch_support}, we now state and prove a more general variant of Lemma \ref{lem::adaptive_tile}. As notation, recall that we assume that the $m$th batch $B_m = b_m(\hat\epsilon_{(1)}, \dots, \hat\epsilon_{(m-1)}, U_m)$ and group $G_m = g_m(\hat\epsilon_{(1)}, \dots, \hat\epsilon_{(m-1)}, U_m)$ are (potentially randomized) functions of the residuals from the first $m-1$ tiles, where $U_m \iid \Unif(0,1)$ are independent uniform noise.

% \begingroup
% \def\thetheorem{\ref{lem::adaptive_tile}}
\begin{lemma}\label{lem::adaptive_tiling_appendix} Suppose Assumptions \ref{assump::batched_exposure} and \ref{assump::localexch_support} hold. Furthermore, for each $m=1,\dots,M$, $(B_m, G_m)$ are random but do not depend on the order of the observations within each of the first $m-1$ tiles. Formally, we assume that for any permutation matrices $P_j \in \R^{|B_j| \times |B_j|}$, 
\begin{equation*}
    b_m(\hat\epsilon_{(1)}, \dots, \hat\epsilon_{(m-1)}, U_m) 
    =
    b_m(P_1 \hat\epsilon_{(1)}, \dots, P_{m-1} \hat\epsilon_{(m-1)}, U_m),
\end{equation*}
and the same holds when replacing $b_m$ with $g_m$. Then $p_{\mathrm{val}}$ in Eq. \ref{eq::hatepspval2} is a valid p-value testing $\mcH_0$.

\begin{proof} It suffices to show that Equation \ref{eq::mainstepone} holds conditionally on the choice of tiles $\{(B_m, G_m)\}_{m \in [M]}$: 
\begin{equation}\label{eq::cond_local_exch}
    (P_1 \epsilon_{(1)}, \dots, P_M \epsilon_{(M)}) \disteq (\epsilon_{(1)}, \dots, \epsilon_{(M)}) \mid \{(B_m, G_m)\}_{m \in [M]}.
\end{equation}
After showing this, the original proof of Theorem \ref{thm::main} will go through after conditioning on the tiles.

The proof is simple, but the notation is subtle. To ease comprehension, recall that by definition $\epsilon_{(m)} = \epsilon_{B_m, G_m}$ where $B_m$ and $G_m$ are random. In this proof, we will use the notation $\epsilon_{B_m, G_m}$ instead of $\epsilon_{(m)}$ to make the dependence on $B_m$ and $G_m$ explicit. 

Let $\mcT = \{(B_m, G_m)\}_{m \in [M]}$ denote the choice of tiles and let $\tau = \{(\beta_m, \gamma_m)\}_{m \in [M]}$ denote an arbitrary \textit{fixed} tiling in the support of $\mcT$. It suffices to show that for any fixed permutation matrices $P_1 \in \R^{|\beta_1| \times |\beta_1|}, \dots, P_M \in \R^{|\beta_M| \times |\beta_M|}$,
\begin{equation}\label{eq::sufficient_for_lemma}
    (P_1 \epsilon_{\beta_1, \gamma_1}, \dots, P_M \epsilon_{\beta_M, \gamma_M}) \disteq (\epsilon_{\beta_1, \gamma_1}, \dots, \epsilon_{\beta_M, \gamma_M}) \mid \mcT = \tau.
\end{equation}
To show this, we note that Assumption \ref{assump::localexch_support} yields the \textit{marginal} result that
\begin{equation}\label{eq::marg_result}
    (P_1 \epsilon_{\beta_1, \gamma_1}, \dots, P_M \epsilon_{\beta_M, \gamma_M}) \disteq (\epsilon_{\beta_1, \gamma_1}, \dots, \epsilon_{\beta_M, \gamma_M}).
\end{equation}
To convert this to a conditional result, note that since $\mcT$ is a function of $\hat\epsilon$ (which itself is a deterministic function of $\epsilon$), we can thus write $\mcT(\epsilon_{\beta_1, \gamma_1}, \dots, \epsilon_{\beta_M, \gamma_M})$ as some function of $\epsilon_{\beta_1, \gamma_1}, \dots, \epsilon_{\beta_M, \gamma_M}$. Using this fact, we define $\mcTperm \defeq \mcT(P_1 \epsilon_{\beta_1, \gamma_1}, \dots, P_M \epsilon_{\beta_M, \gamma_M})$ to be equal to the tiling we \textit{would} have chosen based on the permuted residuals $(P_1 \epsilon_{\beta_1, \gamma_1}, \dots, P_M \epsilon_{\beta_M, \gamma_M})$. Eq. \ref{eq::marg_result} now directly implies that
\begin{equation}
    \left[(P_1 \epsilon_{\beta_1, \gamma_1}, \dots, P_M \epsilon_{\beta_M, \gamma_M}), \I(\mcTperm = \tau)\right]
    \disteq 
    \left[(\epsilon_{\beta_1, \gamma_1}, \dots, \epsilon_{\beta_M, \gamma_M}), \I(\mcT = \tau)\right].
\end{equation}
However, by assumption, $\mcT$ does not depend on the order of rows within each tile defined by $\{(B_m, G_m)\}_{m \in [M]}$. Thus, whenever $\mcT = \tau$, we have that $\{(\beta_m, \gamma_m)\}_{m \in [M]} = \{(B_m, G_m)\}_{m \in [M]}$ and thus $\mcT = \mcTperm$, since $\mcTperm$ is just the value of $\mcT$ after permuting the rows of each tile defined by $\{(\beta_m, \gamma_m)\}_{m \in [M]}$. Thus, $\mcT = \tau$ if and only if $\mcTperm = \tau$, so $\I(\mcT = \tau) = \I(\mcTperm = \tau)$. Combining this with the previous result yields that
\begin{equation}
    \left[(P_1 \epsilon_{\beta_1, \gamma_1}, \dots, P_M \epsilon_{\beta_M, \gamma_M}), \I(\mcT = \tau)\right]
    \disteq 
    \left[(\epsilon_{\beta_1, \gamma_1}, \dots, \epsilon_{\beta_M, \gamma_M}), \I(\mcT = \tau)\right].
\end{equation}
This immediately implies that Eq. \ref{eq::sufficient_for_lemma} holds, since if $(X,Z) \disteq (Y,Z)$ for any random variables $(X,Y,Z)$, then the conditional distributions $X \mid Z \disteq Y \mid Z$ must be equal as well. This concludes the proof.
\end{proof}
\end{lemma}

\subsection{Technical proofs}

\begin{lemma}\label{lem::simple_exch_result} Suppose $Z = (Z_1, \dots, Z_n)$ are exchangeable random variables. Let $\pi_1, \dots, \pi_K : [n] \to [n]$ be random permutations, sampled uniformly at random. Let $Y_k = (Z_{\pi_k(1)}, \dots, Z_{\pi_k(n)})$ for $k \in [K]$ and set $Y_0 = (Z_1, \dots, Z_n)$. Then $(Y_0, Y_1, \dots, Y_K)$ are exchangeable.
\begin{proof} Fix any permutation $\tau : \{0, \dots, K\} \to \{0, \dots, K\}$. It suffices to show that
\begin{equation*}
    (Y_0, Y_1, \dots, Y_K) \disteq (Y_{\tau(0)}, Y_{\tau(1)}, \dots, Y_{\tau(K)}).
\end{equation*}
As notation, let $P_k$ be the permutation matrix such that $P_k Z = Y_k$, for $k \in \{0, \dots, K\}$. Let $J \defeq P_{\tau(0)}^{-1}$. The first step is to observe that
\begin{equation*}
    (Y_{\tau(0)}, Y_{\tau(1)}, \dots, Y_{\tau(K)}) \defeq (P_{\tau(0)} Z, \dots, P_{\tau(K)} Z) \disteq (P_{\tau(0)} J Z, \dots, P_{\tau(K)} J Z),
\end{equation*}
where the distributional equality holds conditional on $\{P_k\}_{k \in [K]}$ since $Z \disteq JZ$ is exchangeable. The second step is to observe that by construction, $P_{\tau(0)} J = P_0$ is the identity permutation. Since $\{P_k\}_{k \in [K]}$ are uniformly random permutations, we know
\begin{equation}\label{eq::permutation_fact}
    (P_{\tau(0)} J, \dots, P_{\tau(K)} J) \disteq (P_0, \dots, P_K).
\end{equation}
Combining the two previous results yields
\begin{equation*}
    (Y_{\tau(0)}, \dots, Y_{\tau(K)}) \disteq (P_{\tau(0)} J Z, \dots, P_{\tau(K)} J Z) \disteq (P_{0} Z, \dots, P_{K} Z) = (Y_0, \dots, Y_K)
\end{equation*}
where the second-to-last equality holds conditional on $Z$.
\end{proof}
    
\end{lemma}

\begin{lemma}\label{lem::technical_exch_result} Using the notation and assumptions from Theorem \ref{thm::main}, the following holds: 
\begin{equation*}
    \left[(P_1^{(r)} \epsilon_{(1)}, \dots, P_M^{(r)} \epsilon_{(M)})\right]_{r=0}^R \text{ are exchangeable.}
\end{equation*}
\begin{proof} To show this, we will show that for any fixed permutation $\pi : \{0, \dots, R\} \to \{0, \dots, R\}$,
\begin{equation}
    \left[(P_1^{(r)} \epsilon_{(1)}, \dots, P_M^{(r)} \epsilon_{(M)})\right]_{r=0}^R \disteq 
    \left[(P_1^{(\pi(r))} \epsilon_{(1)}, \dots, P_M^{(\pi(r))} \epsilon_{(M)})\right]_{r=0}^R 
\end{equation} 

We will show the result in three steps.

\begin{remark} Our proof essentially follows the proof of Lemma \ref{lem::simple_exch_result}, except we apply this argument simultaneously to each tile $m \in [M]$.
\end{remark}

\underline{Step 1}: Let $\Pi_m$ denote the (random) inverse of $P_m^{\pi(0)}$. Equation \ref{eq::mainstepone} (from Step 1 of the proof of Theorem \ref{thm::main}) implies that
\begin{equation}\label{eq::substep1}
    (\epsilon_{(1)}, \dots, \epsilon_{(M)})
    \disteq 
    (\Pi_1 \epsilon_{(1)}, \dots,  \Pi_M \epsilon_{(M)}),
\end{equation}
where in particular, this holds conditional on $\Pi_1, \dots, \Pi_M$ since these are just permutation matrices, and thus it holds unconditionally as well. 

\underline{Step 2}: The previous observation implies that
\begin{equation}\label{eq::substep2}
    \left[(P_1^{(\pi(r))} \epsilon_{(1)}, \dots, P_M^{(\pi(r))} \epsilon_{(M)})\right]_{r=0}^R
    \disteq 
    \left[(P_1^{(\pi(r))} \Pi_1 \epsilon_{(1)}, \dots, P_M^{(\pi(r))} \Pi_M \epsilon_{(M)})\right]_{r=0}^R
\end{equation}
where the above equation holds conditional on all of the random permutation matrices $\{P_m^{(r)}\}_{m \in [M], r \in [R]}$; in particular, it follows by applying identical permutation matrices to both sides of Eq. \ref{eq::substep1}. 

\underline{Step 3}: Third, we observe that
\begin{equation}\label{eq::substep3}
    \left[(P_1^{(r)}, \dots, P_M^{(r)})\right]_{r=0}^R
    \disteq 
    \left[(P_1^{(\pi(r))} \Pi_1, \dots, P_M^{(\pi(r))} \Pi_M)\right]_{r=0}^R
\end{equation}
To see this, it suffices to show $[P_m^{(r)}]_{r=0}^R \disteq [P_m^{(\pi(r))} \Pi_m]_{r=0}^R$ holds for a single fixed $m$, since the randomness for each $m \in [M]$ is completely independent. However, this is exactly the content of Eq. \ref{eq::permutation_fact}.
%Of course, when $r=0$, $P_m^{\pi(r)} \Pi_m = I_{|B_m|} = P_m^{(0)}$ holds deterministically since $\Pi_m$ is by definition the inverse of $P_m^{(r)}$, so it suffices to show that $[P_m^{(r)}]_{r=1}^R \disteq [P_m^{\pi(r)} \Pi_m]_{r=1}^R$. 

Combining steps two and three, we conclude:
\begin{align*}
        \left[(P_1^{(\pi(r))} \epsilon_{(1)}, \dots, P_M^{(\pi(r))} \epsilon_{(M)})\right]_{r=0}^R
    &\disteq 
        \left[(P_1^{(\pi(r))} \Pi_1 \epsilon_{(1)}, \dots, P_M^{(\pi(r))} \Pi_M \epsilon_{(M)})\right]_{r=0}^R \\
    &\disteq 
        \left[(P_1^{(r)} \epsilon_{(1)}, \dots, P_M^{(r)}\epsilon_{(M)})\right]_{r=0}^R
\end{align*}
where in particular, the first line is a restatement of Step 2, and Step 3 proves that the second equality holds conditional on $\epsilon$. This completes the proof. 

\end{proof}
\end{lemma}

\section{Additional details for the empirical analysis}\label{appendix::realdata}

\subsection{Methodological details for the empirical application}\label{subsec::addnmethoddetails}

In this section, we note a few additional details of our empirical analysis.

\underline{Duplicate assets}: Our analyses include all relevant assets from the BFRE model, with two exceptions:
\begin{itemize}[noitemsep, topsep=0pt, leftmargin=*]
    \item We remove exact and near-exact duplicates. I.e., in the analysis of the tech sector, we exclude ``Alphabet class C" from the analysis since ``Alphabet class A" also appears in the dataset.
    \item We exclude partner and subsidiary corporations, e.g., Valero Midstream Partners and Phillips 66 Partners. The parent corporations (e.g., Valero Energy and Phillips 66) are included in our analysis. Note that this preprocessing step only affects the financial and midstream Energy sectors (labeled FIN and EGYOGINT in the BFRE model).\footnote{Although corporations whose names contain the word ``partner" appear in other sectors---for example, ``Surgery Partners" appears in the healthcare sector---we do not exclude such assets since they are not true subsidiary/partner corporations.} 
\end{itemize}

\underline{Nonexistent assets}: There is almost no missing data in our dataset, with one exception: some assets did not exist for parts of our analysis. For example, ``Oak Street Health" was not publicly traded until mid-2020. Our analysis accounts for this by (1) using all available assets to compute maximally precise mosaic/OLS residuals but (2) only computing test statistics based on assets that existed for the entire analysis:

\begin{itemize}[itemsep=0.5pt, topsep=0pt, leftmargin=*]
    \item \textit{Computing residuals}: At time $t$, we include asset $j$ in the cross-sectional regression(s) used to compute mosaic or OLS residuals if and only if asset $j$'s returns are available for the entire week containing time $t$. This guarantees that under Assumption \ref{assump::localexch}, the residuals within each week are locally exchangeable, because they are computed using the same cross-sectional regressions for the entire week. Equivalently, it ensures that no tile contains any missing data. 
    \item \textit{Computing test statistics}: We compute test statistics based only on assets that are available for the entire duration of our analysis. We make this choice since (1) it avoids imputing the values of any missing residuals and (2) it prevents the test statistic from becoming too noisy if an asset has only been available for a small period. On average, at any time $t$, this means that our test statistic is based only on $\approx 58\%$ of the available assets. Thus, our results should only be interpreted as results about this subset of assets. The good news is that the largest and most economically significant assets (e.g. Google, Exxon Mobil, etc) have no missing data and are always included in the test statistic.
    
    The exception is that Figures \ref{fig::r2_plot} and \ref{fig::r2_othersectors} (i) span a slightly longer time frame (since they use five years of training data) and (ii) use test statistics that adaptively search over sparse subsets of assets. Thus, for these plots only, we include all assets in the analysis (imputing any missing residuals as zero) since in principle, the test statistic can simply ``learn" to ignore assets with many missing values if that leads to high power. However, as discussed below, we emphasize that this approach yields valid p-values even if the zero imputation is arbitrarily inaccurate. 
    
\end{itemize}  
Corollary \ref{cor::missingdata} shows that as long as no tile contains any missing data, any choice of test statistic yields a valid p-value under the assumption that each asset's residuals are locally exchangeable (Assumption \ref{assump::localexch}) conditional on the missing data pattern. Intuitively, this is because Theorem \ref{thm::main} follows from the following result (see Step 3 of its proof):
\begin{equation}
    \left\{\left(\tilde\epsilon_{(1)}^{(r)}, \dots, \tilde\epsilon_{(M)}^{(r)}\right)\right\}_{r=0}^R \text{ are exchangeable,}
\end{equation}
where $\tilde\epsilon_{(m)}^{(r)}$ denotes the $m$th tile of the $r$th permuted residual matrix, and by convention $\tilde\epsilon^{(0)} = \hat \epsilon$ denotes the original (unpermuted) mosaic residual matrix. Even when some data are missing, the method above guarantees that no tile contains any missing data; thus, Theorem \ref{thm::main} applies. Please see the proof of Corollary \ref{cor::missingdata} for more details.

\subsection{Sensitivity to the window size}\label{appendix::windowsize}

Throughout the main text, we computed test statistics over time using a sliding window of size $350$ observations. Although this is a somewhat arbitrary choice, varying the window size does not substantially change the results. In particular, Figure \ref{fig::window_sensitivity} repeats the analysis from Figure \ref{fig::motivation_soln} and shows that the shapes of the curves plotting the mosaic test statistics and the mosaic permutation quantiles do not significantly change across different window sizes.

\begin{figure}
    \includegraphics[width=\linewidth]{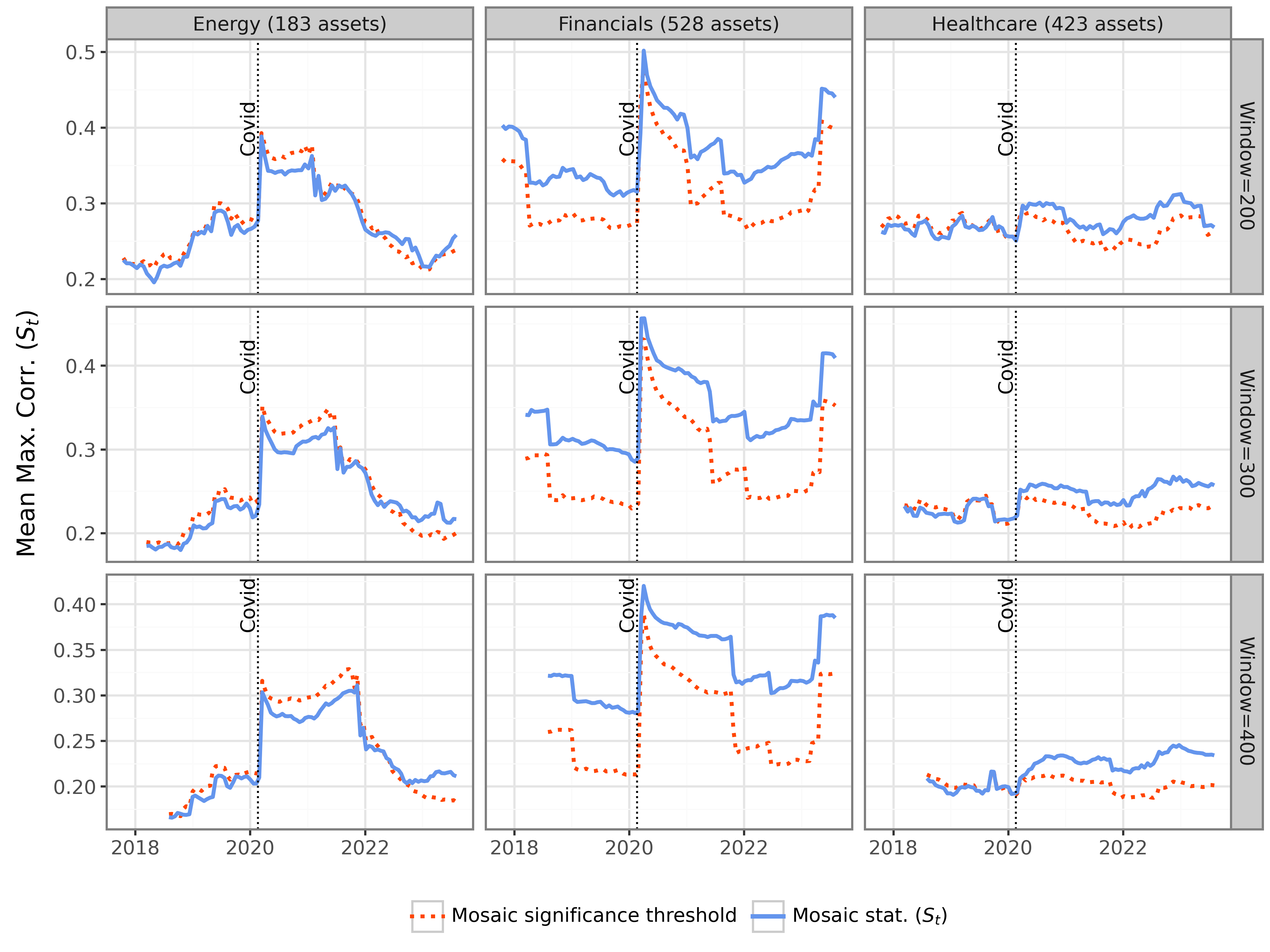}
    \caption{This figure exactly replicates the analysis from Figure \ref{fig::motivation_soln}, except that Figure \ref{fig::motivation_soln} uses a window size of $350$ and this plot varies the window size. In particular, for three industries, it shows the mosaic MMC statistic (plotted every two weeks) and its permutation quantile computed using different sliding window sizes. Note that the shapes of the curves and the relative values of the statistics and their permutation quantiles do not change substantially with different window sizes.}\label{fig::window_sensitivity}
\end{figure}

\subsection{Results for additional sectors}\label{appendix::othersectors}

We now show results for six additional sectors beyond the three from the main text: Consumer Discretionary, Consumer Staples, Industrials, Materials, Tech, and Utilities. In short, we find similar results to Section \ref{sec::realdata}: either (i) we do not consistently reject the null or (ii) we can detect violations of the null $\mcH_0$, but the effect size is too small for us to improve the model consistently. The exception is the consumer discretionary sector, where we can persistently improve the model (albeit by a small degree).

First, Figure \ref{fig::mmc_othersectors} replicates the analysis from Figure \ref{fig::motivation_soln} but for these additional sectors---in particular, it shows the mosaic MMC statistic computed in a sliding window of $350$ observations for each sector as well as the $95\%$ quantile of the permutation distribution. We do not consistently reject the null in the utilities, materials, and consumer staples sectors, although generally we reject the null in the materials sector before COVID and after $2022$. Of course, these sectors do not have many assets, so it is possible this is due to a lack of power. In the consumer discretionary, industrial, and tech sectors, we consistently reject the null. 

Figure \ref{fig::mmc_othersectors} also shows results for two subsectors---real estate (a subsector of finance) and software (a subsector of tech)---of comparable size to the energy sector. We find uniformly significant evidence against the null in real estate, and we find significant evidence against the null in software post-COVID. This analysis demonstrates that our results in the energy sector (in Section \ref{sec::realdata}) are not explained by the relatively small size of the energy sector. 
 
Next, Figure \ref{fig::r2_othersectors} replicates the analysis from Figure \ref{fig::r2_plot} for the six additional industries. I.e., it shows the value of the maximum bi-cross validation $R^2$ test statistic as described in Section \ref{subsec::improvement} and the significance of this test statistic (computed using the mosaic permutation test). As in Section \ref{subsec::improvement}, in all sectors except one, we find that the test statistic is not consistently positive even when it is statistically significant. In other words, we may have power to detect violations of $\mcH_0$, but the effect size is sufficiently small that we cannot consistently improve the model (see Section \ref{subsec::improvement} for more discussion of this phenomenon, which is predicted by several high-dimensional asymptotic theories). The exception is the consumer discretionary sector, where the maximum bi-cross $R^2$ is nearly always positive. In contrast, if we perform this analysis after removing the style factors (as shown by Figure \ref{fig::r2_othersectors}), the p-values become uniformly highly significant, and the maximum bi-cross validation test statistics become larger. Overall, this supports the main finding from the main text, which is that the BFRE model does not fit perfectly but nonetheless explains the most significant correlations among asset returns.

\begin{figure}
    \includegraphics[width=\linewidth]{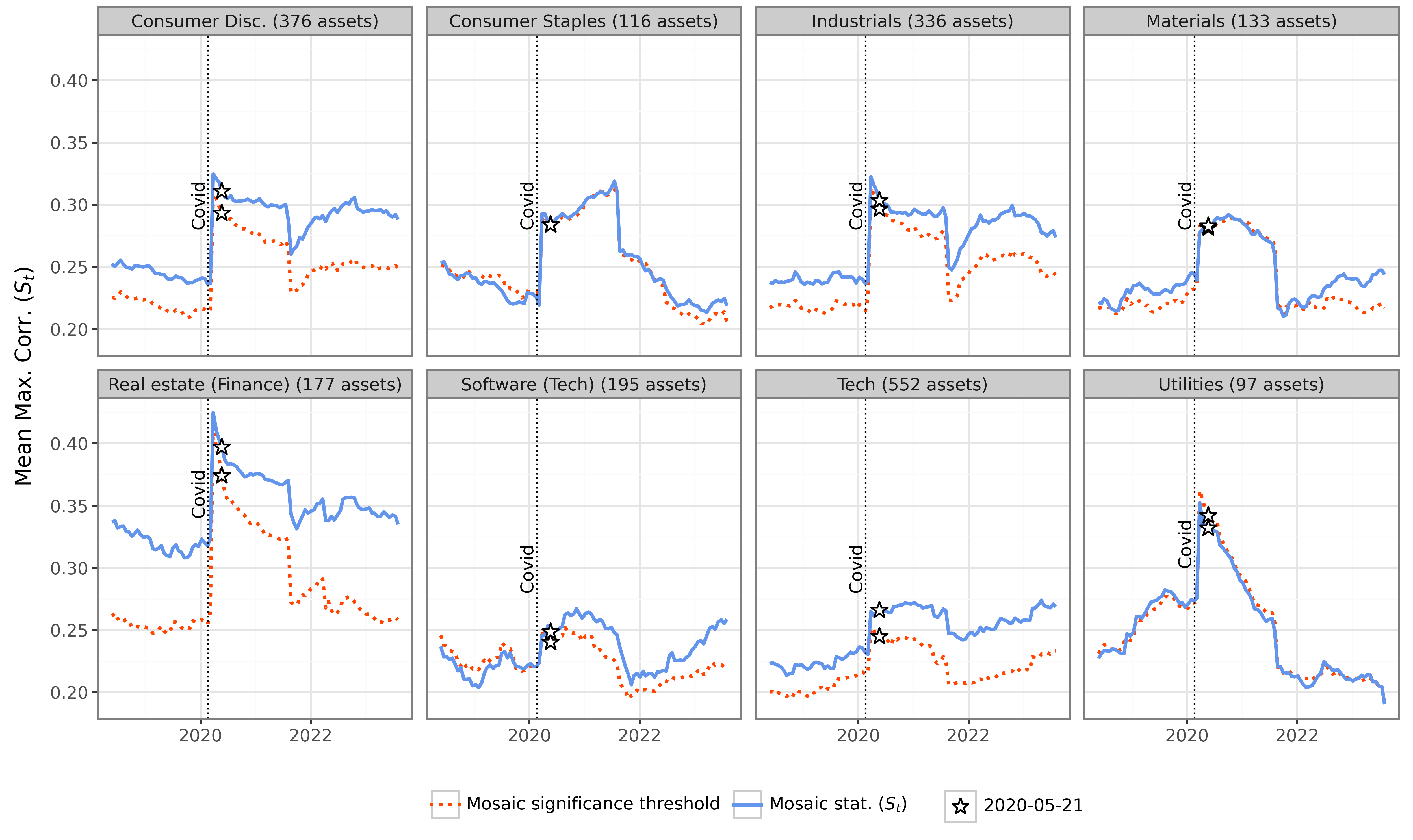}
    \caption{This figure replicates the analysis from Figure \ref{fig::motivation_soln} but for six additional industries, plus two subindustries (real estate and software) of comparable size to the energy sector. For each industry, it shows the mosaic MMC statistic computed in a sliding window of $350$ observations as well as the significance threshold at $\alpha=0.05$.}\label{fig::mmc_othersectors}
\end{figure}

\begin{figure}
    \includegraphics[width=\linewidth]{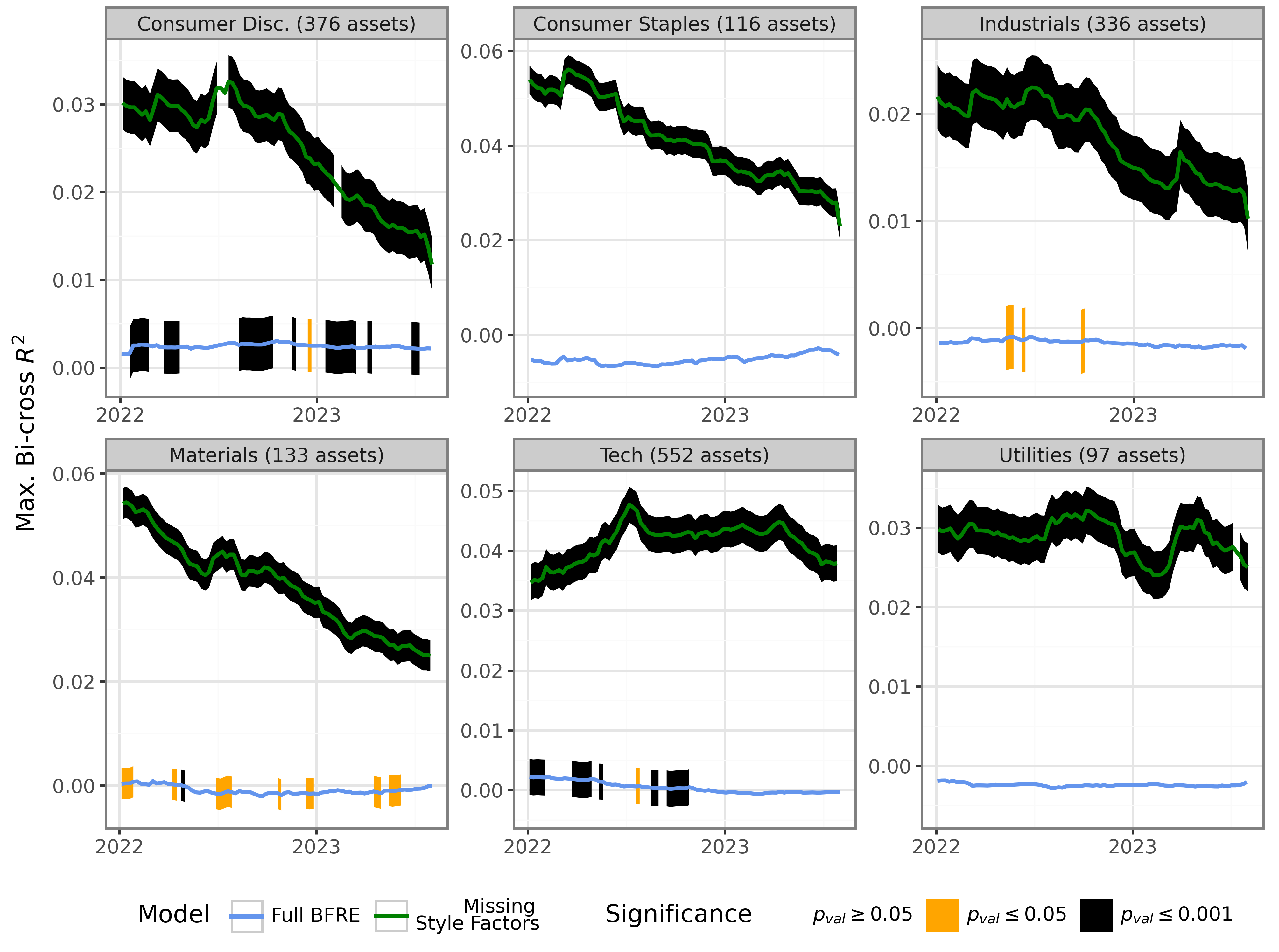}
    \caption{This figure replicates the analysis from Figure \ref{fig::r2_plot} for six additional industries, both for the full BFRE model and for an ablation study that removes the style factors from the BFRE model before performing the analysis. The shadings around each curve denote statistical significance.}\label{fig::r2_othersectors}
\end{figure}

\newpage

\section{Additional methodological details}

\subsection{A greedy algorithm to adaptively choose good tiles}\label{appendix::greedy_adaptive_tiles}

In Section \ref{subsec::adpativetiling}, we introduced a general framework that allows the analyst to adaptively choose the tiling. One method we suggest involves approximately solving the following optimization problem:
\begin{equation}\label{eq::estg_appendix}
    G_{1}, \dots, G_{D} = \argmax_{G_{1}, \dots, G_{D} \text{ partition } [p]} \sum_{j,j' \in [p]} |\hat \Sigma_{j,j'}| \cdot \I(j, j' \text{ are in different groups}).
\end{equation}
Although it is unclear how to exactly solve this optimization problem, Algorithm \ref{alg::greedy_tiles} details an approximate greedy algorithm to solve Eq. \ref{eq::estg_appendix}. In a nutshell, Algorithm \ref{alg::greedy_tiles} randomly initializes $G_1, \dots, G_D$ to each contain one unique asset, and then iteratively adds each unassigned asset to the group which is minimally correlated with that asset. This is essentially a hierarchical ``anti-clustering," since the goal is to \textit{separate} highly correlated assets. Note that although this algorithm may not exactly maximize the objective in Eq. \ref{eq::estg_appendix}, the p-value from Eq. \ref{eq::epspvalv2} will still be valid.

\begin{algorithm}[!h]
\caption{Greedy algorithm to approximately solve Eq. \ref{eq::estg_appendix}}\label{alg::greedy_tiles}
\textbf{Inputs}: Estimated covariance matrix $\hat \Sigma \in \R^{p \times p}$, partition size $D \le p$.
\begin{steps}[topsep=0pt, itemsep=0.5pt, leftmargin=*]
    \item Initialize $G_1 = \{j_1\}, G_2 = \{j_2\}, \dots, G_D = \{j_D\}$ for unique randomly chosen assets $j_1, \dots, j_D \in [p]$.
    \item While $G_1 \cup \dots \cup G_D \ne [p]$ do the following:
    \begin{itemize}[leftmargin=*]
        \item Randomly sample an element $j\opt \in  [p] \setminus (G_1 \cup \dots \cup G_D)$.
        \item Let $d\opt = \argmin_{d \in [D]} \max_{j \in G_d} |\hat\Sigma_{j\opt,j}|$ denote the index of the group $G_d$ which minimizes the maximum absolute estimated correlation between asset $j\opt$ and any asset $j \in G_d$.
        \item Reset $G_{d\opt} = G_{d\opt} \cup \{j\opt\}.$
    \end{itemize}
\end{steps}
\textbf{Return}: Partition $G_1, \dots, G_D$.
\end{algorithm}

\subsection{Details for the sparse PCA algorithm}\label{appendix::greedy_sparse_pca}

We now detail the greedy algorithm we used in Section \ref{subsec::improvement} to approximately solve the sparse PCA problem: 
\begin{equation}\label{eq::sparse_pca_appendix}
    \hat v \approx \max_{\|v\|_2 = 1} v^T \hat C v \text{ s.t. } \|v\|_0 \le \ell.
\end{equation}
In particular, let $\hmc_j \defeq \max_{j' \ne j} |\hat C_{j,j'}|$ be the maximum absolute estimated correlation between asset $j$ and asset $j'$. Let $S \subset [p]$ be the subset of indices of $[p]$ corresponding to the assets with the $\ell$ largest values of $\hmc_j$, so $|S| = \ell$. Finally, let $\hat \nu$ denote the top eigenvector of $\hat C_{S,S}$. We then define $\hat v$ as follows:
\begin{equation*}
    \hat v_j = 
    \begin{cases}
        \hat \nu_j & j \in S \\
        0 & \text{ else.}
    \end{cases}
\end{equation*}
In other words, $S$ is the support of $\hat v$, and on $S$, $\hat v$ equals the maximum eigenvalue of $\hat C_{S,S}$. We picked this algorithm because it is conceptually simple and computationally cheap, but we have not explored other algorithms. Of course, the mosaic permutation test could be used in combination with a maximum bi-cross validation $R^2$ statistic based on any sparse PCA algorithm.

\end{document}